\documentclass[letterpaper]{article}
\usepackage[left=2.4cm, top=2.4cm, bottom=2.4cm, right=2.4cm]{geometry}
\usepackage[utf8]{inputenc}
\usepackage[american]{babel}
\usepackage{amsmath,amssymb,amsthm,latexsym,float,epsfig,bbm,amsrefs,mathtools}
\usepackage{todonotes}
\usepackage{enumitem}
\usepackage{graphicx}
\usepackage{svg}
\usepackage{color,soul}
\usepackage{rotating}
\numberwithin{equation}{section}

\def\R{\mathbb{R}}

\def\C{\mathbb{C}}
\def\E{\mathbb{E}}
\def\P{\mathbb{P}}
\def\1{\mathbbm{1}}

\def\F{\mathcal{F}}

\def\Lsp{\mathcal{L}}
\def\calL{\mathcal{L}}

\newcommand{\eqlaw}{\stackrel{\mathcal{L}}{=}}

\newcommand{\Expecs}[1]{\E^{*}\left[\, #1 \, \right]}
\newcommand{\Expec}[1]{\E\left[\, #1 \, \right]}
\newcommand{\Expecnb}[1]{\E #1 }
\newcommand{\Expecsnb}[1]{\E^{*} #1 }

\newcommand{\proc}[1]{\left( #1 \right)_{t \geq 0}}
\newcommand{\procn}[1]{\left( #1 \right)_{n \geq 1}}
\newcommand{\Ind}[1]{\1_{ \{ #1 \} } }
\newcommand{\Ps}{\P^{*}}
\newcommand{\Prob}[1]{\P\left( #1 \right)}
\newcommand{\Probs}[1]{\P^{*}\left( #1 \right)}

\newcommand{\levy}{L\'{e}vy }

\newcommand{\ltrip}{\left(b,0,\nu \right)}

\newcommand{\gdelsof}[1]{\gamma_{\delta}^{*}\left( #1 \right)}
\newcommand{\gdelof}[1]{\gamma_{\delta}\left( #1 \right)}

\def\abs#1{\left\lvert #1 \right\rvert}
\def\rv[#1]#2{RV_{#1}^{#2}}
\DeclareMathOperator{\sgn}{sgn}

\theoremstyle{plain}
\newtheorem{thm}{Theorem}[section]
\newtheorem{lem}[thm]{Lemma}
\newtheorem{prop}[thm]{Proposition}
\newtheorem{cor}[thm]{Corollary}
\theoremstyle{definition}
\newtheorem{defn}[thm]{Definition}

\theoremstyle{remark}
\newtheorem{rem}[thm]{Remark}

\usepackage[colorlinks]{hyperref}
\title{At-the-money short-time call-price asymptotics for new classes of 
    exponential L\'evy models}
\author{Allen Hoffmeyer\thanks{Penumbra Investment Group, New York, NY 10020, USA ({\tt allen.hoffmeyer@gmail.com}).}
\and Christian Houdr\'e\thanks{School of Mathematics, Georgia Institute of 
    Technology, Atlanta, GA 30332, USA ({\tt houdre@math.gatech.edu}).}
   \thanks{Research supported in part by grants \#524678 and MP-TSM-00002660 from the Simons Foundation.}} 

\date{\today}

\begin{document}

\maketitle

\begin{abstract}
We develop at-the-money call-price and implied volatility asymptotic
expansions in time to maturity for a class of asset-price models whose
log returns follow a \levy process. Under mild assumptions placing the
driving \levy process in the small-time domain of attraction
of an $\alpha$-stable law with $\alpha \in (1,2)$, we give first-order
at-the-money call-price and implied volatility asymptotics. A key
observation is that both the stable domain of attraction and the
finiteness of the centering constant $\bar{\mu}$ are preserved under
the share measure transformation, so that all of the distributional
input needed for the call-price expansion can be read off from the
regular variation of the \levy measure near the origin. When the
\levy process has no Brownian component, new rates of convergence
of the
form $t^{1/\alpha} \ell(t)$ where $\ell$ is a slowly varying function are
obtained. We provide an example of an exponential \levy model exhibiting this
behavior, with $\ell$ not asymptotically constant, yielding a
convergence rate of $(t / \log(1/t))^{1/\alpha}$. In the case of a
\levy process with Brownian component, we show that the jump
contribution is always lower order, so that the leading $\sqrt{t}$
behavior of the at-the-money call price is universal and driven
entirely by the Gaussian part of the characteristic triplet.

\vspace{0.2 cm}

\vspace{0.1 cm}
\end{abstract}

\section{Introduction}

The Black--Scholes paradigm is analytically elegant yet empirically 
inadequate: implied volatility is not flat across strikes or maturities and 
log returns are not Gaussian. In practice, smiles and skews persist even at 
short maturities, and tail risk is materially underpriced when returns are 
modeled by a diffusion with exponentially decaying tails. Numerous extensions 
have been proposed. Local-volatility models \cite{dupire_pricing} reproduce today’s 
smile by design but are often unstable out of sample; 
stochastic-volatility models (e.g.\ Heston/SABR) capture dynamics but 
typically miss the exact near-expiry smile without substantial calibration 
effort. In this work we focus on exponential \levy models, where the log-price 
is a \levy process and jumps produce the empirically observed smile and 
heavier tails while preserving analytic tractability in many cases (see 
the surveys and references in \cites{tankov_review,boy_leven,figueroa_forde}).

A major thread of research over the last two decades studies the 
\emph{small-time} (near-expiry) behavior of option prices and implied 
volatility in exponential \levy models. Early milestones include general 
at-the-money (ATM) price asymptotics and links to first absolute moments 
\cite{karbe_nutz}, small-time smile behavior in stable-like and tempered 
settings \cites{figueroa_forde,tankov_review}, and model-specific expansions 
for CGMY and related classes \cites{fl_houdre_cgmy,figueroa_lopez_gong_houdre_2014}. Since 
then, the literature has clarified how jump activity (e.g.\ the 
Blumenthal--Getoor index) and any Brownian component jointly determine ATM 
scalings and higher-order terms. Notably, \cite{mijatovic_tankov_2016} give 
a unifying view of short-maturity implied volatility in exponential \levy 
models with jumps, \cite{figueroa_lopez_gong_houdre_2014} derive high-order ATM price 
expansions for broad tempered-stable classes, 
\cite{figueroa2012small} obtain small-time expansions for 
distributions, densities, and option prices in stochastic volatility models 
with \levy jumps, and \cite{figueroa_lopez_olafsson_2016} 
extend close-to-the-money results in the presence of stochastic volatility.
Additional related developments 
include model-engineering and calibration perspectives for exponential 
\levy surfaces \cite{andersen_lipton}, estimation advances for 
tempered-stable models of infinite variation \cite{figueroa_lopez_gong_han_2022},
and short-time implied-volatility limits for additive tempered-stable and 
allied processes \cite{azzone_2024}. (For comparison with the diffusion 
literature on near-expiry ATM behavior, see e.g.\ \cite{el_euch_rosenbaum_2019} 
and references therein.)

Below, we develop near-expiry asymptotic expansions for ATM call prices and 
the corresponding implied volatility under exponential \levy dynamics, with 
an emphasis on stable domains of attraction and regular variation. Under mild 
structural assumptions placing the driving \levy process in the domain of 
attraction of a (possibly asymmetric) stable law, we obtain first-order 
ATM expansions that unify and streamline existing results while covering 
cases beyond the standard tempered-stable families. The formulation 
highlights the role of regular variation of the \levy measure near the 
origin. 

This article is organized as follows.
Section~\ref{sec:doa_rv} develops the domain-of-attraction and 
regular-variation framework needed for the sequel, and establishes that both 
the stable domain of attraction and the finiteness of the centering constant 
are preserved under the share measure transformation.
Section~\ref{sec:atm_asymptotics} states and proves the main first-order 
ATM call-price and implied volatility asymptotics, treating the pure-jump 
and Brownian-component cases separately.
Section~\ref{sec:example} presents an explicit exponential \levy model whose 
first-order convergence rate involves a logarithmic correction to the classical scaling.
Section~\ref{sec:conclusion} concludes and discusses directions for future 
work. Appendix~A collects the regular-variation background used throughout, 
and Appendix~B contains the proofs.

We start by briefly recalling the basic material on \levy processes and exponential-\levy 
models used in what follows; see \cites{sato,applebaum,cont_tankov} for comprehensive accounts.

A stochastic process $\proc{X_t}$ on $(\Omega,\F,\P)$ with values in $\R$ is a 
\emph{\levy process} if it has stationary and independent increments, càdlàg paths, 
and $X_0=0$ a.s.  Every \levy process is uniquely characterized 
by its \emph{characteristic triplet} $(b,\sigma,\nu)$, where $b\in\R$, $\sigma\ge0$, and 
$\nu$ is a positive Borel measure on $\left(\R, \mathcal{B}(\R) \right)$ without
atom at the origin satisfying 
$\int_{\scriptscriptstyle \R} (1\wedge x^2)\nu(dx)<\infty$. Hence $\E[e^{iuX_t}] = e^{t\psi(u)}$ where  
the characteristic exponent is
\[
    \psi(u)
      = iub - \tfrac12\sigma^2u^2
        + \int_{\R}
            \big( e^{iux}-1 - iu x \1_{\{|x|\le1\}} \big)\,\nu(dx),
        \qquad u\in\R.
\]
This is compactly written as $X_t\sim \Lsp(b,\sigma,\nu)$.

Under mild conditions, $X_t$ can be written as
\[
  X_t = bt + \sigma W_t
        + \int_{|x|\le1} x\,\tilde N(t,dx)
        + \int_{|x|>1} x\, N(t,dx),
\]
where $W$ is a Brownian motion, $N$ is a Poisson random measure with intensity $\nu(dx)\,dt$, 
and $\tilde N$ is its compensated version. 

If $\Ps$ is an equivalent measure defined by
\[
    \frac{d\Ps}{d\P}\Big|_{\F_t}
      = \exp\!\big(\theta X_t - t\psi(-i\theta)\big),
\]
then $\proc{X_t}$ remains a \levy process under $\Ps$, with transformed triplet
$(b_\theta,\sigma,\nu_\theta)$ given by the Esscher transform
$\nu_\theta(dx)=e^{\theta x}\nu(dx)$ and $b_\theta=b+\sigma^2\theta
+\int_{|x|\le1}x(e^{\theta x}-1)\nu(dx)$.  Such changes of measure appear naturally in 
risk-neutral valuation.

In an exponential \levy model the asset price is
\[
  S_t = S_0 e^{X_t}, \qquad X_t\sim\Lsp(b,\sigma,\nu),
\]
so that $\E[S_t]=S_0 e^{t\psi(-i)}$, where $\psi$ is the L\'evy--Khintchine
exponent associated with the triplet $(b,\sigma,\nu)$ evaluated at $-i$.  
Since $\psi$ is a 
priori defined only for real arguments, evaluating it at $-i$ requires 
justification.  For the exponential moment $\Expec{e^{X_t}}$ to be finite it is 
necessary and sufficient that
\begin{equation}
  \int_{\{|y|>1\}} e^{y}\,\nu(dy) < \infty, \label{eq:levy_model_moment_condition}
\end{equation}
which we assume throughout.  Under this condition, the L\'evy--Khintchine 
integral converges when $u$ is replaced by $u+iz$ for any $z\in[0,1]$, so 
$\psi$ extends analytically to the strip 
$\{u\in\C:-1\le\Im(u)\le 0\}$; in particular $-i$ lies in this strip and 
$\psi(-i)$ is well defined (see \cite{sato}, Theorem 25.17). Under a
risk-neutral measure the discounted asset $e^{-rt}S_t$ must be a martingale,
so $\E[e^{X_t}]=e^{rt}$ and hence $\psi(-i)=r$.  In this paper we work in the
zero-interest (or forward) case $r=0$, so that $\psi(-i)=0$ and the drift
parameter is fixed by the martingale condition
\begin{align}
        b = -\frac{\sigma^{2}}{2} - \int_{-\infty}^{\infty} \left( e^{y} 
        -1 - y \1_{\left\{ \abs{y} \leq 1 \right\} } \right) \nu(dy). 
        \label{eq:b_martingale_condition}
    \end{align}
    The above conventions (triplet $(b,\sigma,\nu)$, truncation $x\1_{|x|\le1}$, and 
characteristic exponent $\psi$) are maintained throughout the paper.

The asymptotics of at-the-money option prices and implied volatility are the 
main objects of study in this manuscript. For this purpose, we discuss a few results
that will be necessary in deriving the first-order asymptotics. 

First, we are interested in the short-time behavior of the at-the-money call price
\[
    c(t,0) = \Expecnb{\left( e^{X_{t}} - 1 \right)_{+}},
\]
interpreted under the risk-neutral measure. 
To this end, we use a slightly more convenient representation of the function 
$c$ due to \cite{carr_madan} (see also \cite{figueroa_forde}). 
We work with the share measure $\Ps$ obtained from the Esscher
transform with parameter $\theta=1$ (so that $d\Ps/d\P|_{\F_t}=e^{X_t}$ when
$\psi(-i)=0$), which satisfies, for all Borel sets $D \subset \R$,
\[
    \Probs{D} = \Expecnb{e^{X_{t}} \1_{D}}.
\]
With this approach, \cite{carr_madan} showed the following.

\begin{thm}
    \label{thm:carr_madan}
    Under $\Ps$, let $E$ be a mean $1$ exponential random variable that is independent of 
    $\proc{X_{t}}$. Then, 
    \begin{align}
        \frac{1}{S_{0}} \Expecnb{\left( S_{t} - K \right)_{+}} =
        \Probs{X_{t} - E > \log{\left( \frac{K}{S_{0}} \right)}}.
        \label{eq:option_formula}
    \end{align}
\end{thm}

\begin{cor}
    Letting $K=S_0$, the normalized, at-the-money European call option price 
    has representation
    \begin{align}
        c(t,0) &= \frac{1}{S_{0}} \Expecnb{\left( S_{t} - S_0 \right)_{+}} 
        = \int_{0}^{\infty} e^{-x} \Probs{X_{t} \geq x} dx.
        \label{eq:cm_rep}
    \end{align}
\end{cor}

These last two results will be used to derive the first-order asymptotic behavior of 
$c(t,0)$ as $t \downarrow 0$ for a wide class of \levy models. 

\section{Stable Domains of Attraction}
\label{sec:doa_rv}

Stable domains of attraction play a central role in the short-maturity
behavior of at-the-money option prices. In much of the existing literature
(see, e.g., \cites{figueroa_forde,fl_houdre_cgmy,figueroa_lopez_gong_houdre_2014,
tankov_review,karbe_nutz}), the underlying \levy process $\proc{X_t}$ is
assumed to satisfy a small-time stable limit of the form
\begin{equation}
    \frac{X_t - A_t}{B_t} \Rightarrow Z, \qquad t \downarrow 0,
    \label{eq:levy_stable_doa_intro}
\end{equation}
where $Z$ is an $\alpha$-stable random variable with $\alpha \in (0,2]$,
and $\Rightarrow$ denotes convergence in distribution. In particular, in the
pure-jump case with $\alpha \in (1,2)$ one typically has $B_t = t^{1/\alpha}$
up to a slowly varying factor. Our goal in this section is to make precise
the regular-variation assumptions under which \eqref{eq:levy_stable_doa_intro}
holds, and to describe the associated scaling $B_t$ in sufficient generality
for the option-pricing applications that follow.

\subsection{Preliminaries}

In this subsection we collect the minimal regular variation notation that will be used 
throughout the rest of the paper. Our goal is to characterize when the small-time rescaling 
\eqref{eq:levy_norm_sum} of a \levy process converges to a non-Gaussian stable law and to 
identify the corresponding scaling $B_t$ in terms of the behavior of the \levy measure 
near the origin. Regular variation provides the natural language for 
describing this local behavior.

A real-valued function $f$ is \emph{regularly varying of index $\rho$ at $\infty$}
if for every $\lambda \in \left(0, \infty \right)$, 
\[
\lim_{x\rightarrow \infty} \frac{f\left(\lambda x \right)}{f(x)} = \lambda^{\rho}.
\]
We denote this as $f \in RV_{\rho}^{\infty}$. Similarly, a real-valued function $\ell$ 
is \emph{slowly varying at $\infty$} if for every $\lambda \in (0,\infty)$, 
\[
\lim_{x\rightarrow \infty} \frac{\ell\left(\lambda x \right)}{\ell(x)} = 1,
\]
and this is denoted by $\ell \in RV_{0}^{\infty}$. 
Finally, a function $f$ is \emph{regularly varying at $0$} 
(from the right) with index $\rho$ if $f\left( \frac{1}{\cdot} \right) 
\in RV_{-\rho}^{\infty}$. We denote this by writing $f \in RV_{\rho}^{0}$.
We will sometimes drop the superscript and write, e.g., $f\in RV_{\rho}$ if it is clear
from context whether the function varies regularly (or slowly) at $\infty$ or $0$.
We write $f(x) \sim g(x)$ as $x \to a$ to mean $\lim_{x \to a} f(x)/g(x) = 1$.

We now introduce some notation for $X_t \sim \calL\left(b, \sigma, \nu\right)$. For $x>0$, let 
\begin{align}
    \gamma\left( x \right) = \gamma_{+}(x) + \gamma_{-}(x) &:= \nu\left( 
    \left\{ y > x \right\} \right) + \nu\left( \left\{ y < -x \right\} \right),
    \label{eq:levy_tail} 
\end{align}
and
\begin{align}
    V(x) &:= \int_{\abs{y} \leq x} y^{2} \nu(dy). \label{eq:var_def}
\end{align}

Here $\gamma$ is the two-sided tail function of the \levy measure and $V$ is 
the truncated second moment of jumps of size at most $x$; in particular, 
$V$ encodes the contribution of small jumps to the quadratic variation of $X$.

Next, we are mainly interested in conditions under which
\begin{align}
    \frac{X_{t} - A_{t}}{B_{t}}
    \label{eq:levy_norm_sum}
\end{align}
converges in distribution to an $\alpha$-stable distribution, 
$\alpha \in (1,2)$, as 
$t\rightarrow 0$ where 
$A: [0,\infty) \rightarrow \R$ 
and $B: (0, \infty) \rightarrow (0,\infty)$ are functions with 
$\lim_{t \rightarrow 0} 
B_{t} = 0$. If the \levy process has finite second moment, 
then the Central Limit Theorem (or the L\'{e}vy--Khintchine formula) gives
\[
    \frac{X_{t} - t \Expecnb{X_{1}}}{\sqrt{t}} \Rightarrow \mathcal{N}(0,
    \sigma^{2}),
\]
as $t \rightarrow 0$. This is formalized as follows:

\begin{defn}
    A stochastic process $\left( X_{t} \right)_{t\geq 0}$ on a probability 
    space $\left( \Omega, \mathcal{F}, \P \right)$ is in the \emph{domain of
    attraction} (DOA) of a stable random variable $Z$ at 
    $a \in \left\{ 0, \infty 
    \right\}$ if there exist functions $A$ and $B$ with $A_{t} \in \R$, 
    $B_{t} > 0$ for every $t \geq 0$ such that
    \begin{align}
        \frac{X_{t} - A_{t}}{B_{t}} \Rightarrow Z, \label{eq:doa}
    \end{align}
    as $t \rightarrow a$. 
\end{defn}
Throughout, we are exclusively concerned with the small-time case $a=0$, so 
all domains of attraction will be at $0$.

\subsection{Stable domains of attraction for \levy processes}

We now focus on domains of attraction in the \levy setting, following 
\cite{maller_mason} and \cite{grabchak}. In particular, we recall the 
conditions from \cite{maller_mason}, which characterize when a \levy process 
is in the domain of attraction of a stable law at $0$. 

\begin{thm} \label{thm:maller_mason}
    The following are equivalent:
    \begin{enumerate}[label=(\roman*)]
        \item there exist real-valued functions $A$ and $B$ with
            $B_{t}>0$ for $t>0$ and $\lim_{t \rightarrow 0} B_{t} = 0$ such that
            \[
                \frac{X_{t} - A_{t}}{B_{t}} \Rightarrow Z,
            \]
            as $t \rightarrow 0$, where $Z$ is an a.s.\ finite, nondegenerate
            random variable (in fact, $Z$ is necessarily an $\alpha$-stable 
            random variable with $\alpha \in (0,2]$);
        \item either (a) the function $\gamma \in RV_{-\alpha}^{0}$ with $\alpha \in 
            (0,2)$ and the limits
            \begin{align}
                p_{\pm} := \lim_{x \downarrow 0} \frac{\gamma_{\pm}(x)}{\gamma(x)},
                \label{eq:tail_limits}
            \end{align}
            exist \\
            or (b) $V$ is slowly varying at $0$. 
    \end{enumerate}
\end{thm}

Note that in $(ii)$ above, exactly one of $(a)$ or $(b)$ holds. 
The case where $V$ is slowly varying at $0$ corresponds 
to the case where the central limit theorem applies with 
$B_{t} = \sqrt{t}$ and 
$Z \sim \mathcal{N}\left( 0, \sigma^{2} \right)$. Also, although the statement 
of the theorem does not require it, the function $B$ can be chosen 
monotone decreasing. Finally, the above theorem shows that the weak 
convergence of expressions like \eqref{eq:doa} is necessarily towards 
an $\alpha$-stable random variable, $0<\alpha\leq 2$, if 
$\left( X_{t} \right)_{t \geq 0}$ is a \levy 
process. Additionally, the authors showed that for $x>0$, defining
\begin{align}
        \mu_{x} &:= b - \int_{x \leq \abs{y} \leq 1} y \nu(dy),
    \label{eq:mu_epsilon} 
\end{align}
and
\begin{align}
    U(x) &:=  2 \int_{0}^{x} y \gamma(y) dy, \label{eq:U_func}
\end{align}
the functions
\[
    B_{t} = \inf\left\{ 0<x \leq 1: x^{-2} U(x) \leq \frac{1}{t} \right\} 
    \mbox{\;\; and \;\;} 
    A_{t} = t \mu_{B_{t}}, 
\]
will do. In particular, $A_{t} = O\left( t \right)$ when 
$\sup_{0<x<1} \abs{\mu_{x}} < \infty$. 

Using Theorem \ref{thm:maller_mason} together with the L\'evy--Khintchine representation, 
we obtain a well-known simple and explicit representation of the limiting stable 
random variable $Z$. Let $p_{\pm} \geq 0$ with $p_+ + p_- = 1$, then for any $s\in \R$,
\begin{align}
        \varphi_{Z}(s) = \exp{\Big( -c_{\alpha} \abs{s}^{\alpha} \left( 1 - i \left(
        p_{+} - p_{-}\right) \sgn{(s)} \tan{\left( \pi \alpha/2 \right)} 
        \right) \Big)}, \label{eq:stable_chfn}
    \end{align}
where $c_{\alpha}>0$ is a constant.

In what follows, we require that the \levy process 
$\proc{X_{t}}$ satisfies $\Expec{\abs{X_{t}} e^{X_{t}}} < \infty $ for all 
$t\geq 0 $, equivalently (see \cite{sato}, Theorem 25.3) that
\begin{align}
    \int_{\abs{x} > 1} \abs{x} e^{x} \nu(dx) < \infty.
    \label{eq:ex_moment}
\end{align}
Additionally, we assume that $\nu$ has a density with respect to Lebesgue measure, 
i.e., that there exists a function $\xi \geq 0$ such that for any Borel set 
$D$, 
\[
    \nu\left( D \right) = \int_{D} \xi(x) dx. 
\] 
Note that $\P^{*}$, which is the natural measure under which at-the-money call prices 
take the form of expectations of $X_t$-functionals, is well-defined by 
\eqref{eq:ex_moment} and Example 33.14 in \cite{sato}. It will be the key measure for our
short-time option price asymptotics. We assume throughout this 
section that $\sigma = 0$, so that $X$ is a pure-jump \levy process. Under 
$\Ps$, the process $\left( X_{t} \right)_{t \geq 0}$ is again
a \levy process with triplet $\left( b^{*}, 0, \nu^{*} \right)$ where
\begin{align}
    \nu^{*}\left( dx \right) = e^{x} \nu(dx) = e^{x} \xi(x) dx,
    \label{eq:nu_star} 
\end{align}
and
\begin{align}
    b^{*} = b + \int_{\abs{x} \leq 1} x \left( \nu^{*} - \nu \right)(dx) = 
    b + \int_{\abs{x} \leq 1} x\left( e^{x} - 1\right) \xi(x)\,dx.
    \label{eq:b_star}
\end{align}

Throughout this section, we define quantities under both 
$\P$ and $\Ps$, using the star notation to denote the 
corresponding quantity under $\Ps$. For instance, 
$\mu_{\varepsilon}$ is defined in \eqref{eq:mu_epsilon} and its 
share-measure counterpart is
\[
    \mu_{\varepsilon}^{*} := b^{*} - \int_{\varepsilon \leq \abs{x} \leq 1} x 
    \,\nu^{*}(dx).
\]

The star-measure analogue of $\gamma$ from \eqref{eq:levy_tail} is
\[
    \gamma^{*}(x) := \int_{\abs{y} > x} \nu^{*}(dy),
\]
for $x>0$.

We now define a few more functions and constants that we will need thereafter. 
For $x>0$, set
\begin{align}
    \xi_{S}(x) := \xi(x) + \xi(-x),
    \label{eq:levy_density_symmetrization}
\end{align}
while 
\begin{align}
    \xi_{S}^{*}(x) := e^{x} \xi(x) + e^{-x} \xi(-x).
    \label{eq:levy_density_symmetrization_star}
\end{align}
We are interested in the quantities
\[
    \bar{\mu} = \sup_{0 < \eta < \infty} \abs{\mu_{\eta}} 
    \mbox{\;\;\;\; and \;\;\;\;}
    \bar{\mu}^{*} = \sup_{0 < \eta < \infty} \abs{\mu_{\eta}^{*}},
\]
which we will need to be finite. Note that 
\[
    \sup_{1 \leq \eta < \infty} \abs{\mu_{\eta}} =  
    \sup_{1 \leq \eta < \infty} \abs{b + \int_{ 1 \leq \abs{x} \leq \eta} 
    x \nu\left( dx \right)} \leq 
    \abs{b} + \int_{ 1 \leq \abs{x} < \infty } 
    \abs{x} \nu\left( dx \right) < \infty,
\]
and
\[
    \sup_{1 \leq \eta < \infty} \abs{\mu_{\eta}^{*}} =  
    \sup_{1 \leq \eta < \infty} \abs{b^{*} + \int_{ 1 \leq \abs{x} \leq \eta} 
    x \nu^{*}\left( dx \right)} \leq 
    \abs{b^{*}} + \int_{ 1 \leq \abs{x} < \infty } 
    \abs{x} \nu^{*}\left( dx \right) < \infty,
\]
since the \levy process $\left( X_{t} \right)_{t \geq 0}$ has finite 
first moment under both $\P$ and $\Ps$.  
So, we need only consider when the quantities 
\[
    \bar{\mu}_{0} = \sup_{0 < \eta \leq 1} \abs{\mu_{\eta}} \mbox{\;\;\;\; and 
    \;\;\;\;}
    \bar{\mu}^{*}_{0} = \sup_{0 < \eta \leq 1} \abs{\mu_{\eta}^{*}},
\]
are finite (e.g. when $\nu$ is symmetric, see \cite{figueroa_forde}).

\subsection{Share measure invariance and the normalizing function}
First, we investigate how the share measure transformation affects the regular 
variation property of the \levy measure. Intuitively, this preservation of 
regular variation stems from the fact that the property depends only on the 
behavior of $\nu$ near the origin, and the transformed \levy measure 
$\nu^{*}(dx) = e^{x}\nu(dx)$ has the same local behavior for $x$ close to $0$ 
because $e^{x} \approx 1$ there.

\begin{prop} \label{lem:preserve_reg_var}
    If $\nu$ is regularly varying of index $\alpha$ at $0$, then 
    $\nu^{*}$ is also regularly varying of index $\alpha$ at $0$. 
\end{prop}

As a major consequence of the proof of Proposition \ref{lem:preserve_reg_var},
the constants $p_{+}$ and $p_{-}$ in \eqref{eq:p_lim_equality} and 
\eqref{eq:q_lim_equality} remain unchanged under the 
measure transform. This fact, combined with \eqref{eq:stable_chfn}, 
gives the following result.

\begin{prop} \label{lem:same_stable_rv}
    Let \eqref{eq:doa} hold with respect to $\P$ and thus with 
    respect to $\Ps$, where $Z$ is an $\alpha$-stable random variable with $\alpha \in (0,2)$. 
    Then $Z$ has the same representation under both $\P$ and
    $\Ps$. That is, the parameters of the stable distribution $Z$ are the 
    same under both probability measures $\P$ and $\Ps$. 
\end{prop}

In summary, the small-time stable domain of attraction of $X$ is completely determined by 
the regular variation of the \levy measure near the origin, and this structure is 
preserved by the share measure change. In particular, the same limiting 
$\alpha$-stable random variable $Z$ governs the fluctuations 
of $X_t$ under both $\P$ and $\P^{*}$, so all of the distributional input needed 
for the at-the-money option price expansions can be read off from the behavior 
of $\nu$ at zero.

Next, we show that the finiteness of the constant $\bar{\mu}$ is also 
a property that survives the share measure transformation. This quantity will
be important for Theorem \ref{thm:main_result}, which is one of the main 
new results of the paper.

\begin{prop} \label{lem:mu_bar_finite}
    $\bar{\mu} < \infty$ if and only if $\bar{\mu}^{*} < \infty$, equivalently
    $\bar{\mu}_{0} < \infty$ if and only if $\bar{\mu}^{*}_{0} < \infty$.
\end{prop}

The previous propositions show that both the stable domain-of-attraction structure
and the finiteness of $\bar{\mu}$ are invariant under the share measure
transformation. In particular, the same strictly $\alpha$-stable random
variable $Z$ (and the same centering constant $\bar{\mu}$ whenever it is finite)
govern the small-time behavior of $X_t$ under both $\P$ and $\Ps$. We are now
in a position to connect this structure to the behavior of the normalizing
function $B_t$ and hence to the short-time asymptotics of the at-the-money call
price.

Previous results on ATM call prices (e.g., \cite{figueroa_forde}, 
\cite{tankov_review}, and \cite{karbe_nutz}) considered only
the case $B_{t} = \kappa t^{1/\alpha}$ whenever $1 < \alpha < 2$, where
$\kappa>0$ is a constant. As we show in 
Section~\ref{sec:atm_asymptotics}, the framework developed here 
yields more general asymptotics for ATM call option prices. 

It is known (see \cite{feller}, \cite{grabchak}, and 
\cite{meerschaert_scheffler}) that the 
rate function $B \in RV_{1/\alpha}^{0}$ whenever the convergence is towards an
$\alpha$-stable random variable. Our goal for the remainder of this 
subsection is to further pin down the behavior of $B$ when 
\eqref{eq:doa} is satisfied. Throughout, we use the notation 
$\beta_{t} := 1 / B_{t}$ for convenience.

Assume that \eqref{eq:doa} holds for the \levy process $\left( X_{t}
\right)_{t \geq 0}$ under the measure $\P$ (and 
hence also under $\Ps$). Thus, the 
\levy measures $\nu$ and $\nu^{*}$ are regularly varying with index $\alpha >0$ 
at $0$, and we further assume that $\alpha \in (1,2)$. Since $\gamma^{*}$ is regularly 
varying at $0$ of order $-\alpha$, it has representation
\begin{align}
    \gamma^{*}\left( x \right) = x^{-\alpha} \ell\left( \frac{1}{x} \right),
    \label{eq:gamma_rv}
\end{align}
for all $x>0$, where $\ell$ is slowly varying at $\infty$ (equivalently, 
we can take $\gamma^{*}(x) = x^{-\alpha} \tilde{\ell}(x)$ where $\tilde{\ell}$ 
is slowly varying at $0$). We deal with $\gamma^{*}$ rather than $\gamma$ as most of
our calculations will be done with the quantities under the measure $\Ps$. 

The representation for $\gamma^{*}$ is derived in the following way. First,
note that $\gamma^{*}\left( \cdot \right) \in RV_{-\alpha}^{0}$ if and only if 
$\gamma^{*}\left( \frac{1}{\cdot} \right) \in RV_{\alpha}^{\infty}$, and then the
representation theorem for regularly varying functions gives 
\eqref{eq:gamma_rv}.

Note that $\ell$ is asymptotically controlled near $\infty$ by its 
slow variation, but its behavior near $0$ is essentially 
unconstrained apart from the monotonicity of $\gamma^{*}$ and the 
integrability condition $\int_{-1}^{1} x^2 \nu^{*}(dx) < \infty$.

We can simplify how to look at \eqref{eq:doa} because we do not need the 
additive correction term. In \cite{maller_mason}, the authors show that $A_{t}$
can be taken to be $O(t)$ (recall that $\alpha \in (1,2)$). So, under
$\Ps$ 
\[
    \beta_{t}\left( X_{t} - A_{t} \right) \Rightarrow Z,
\]
as $t \rightarrow 0$. Recall that $\beta_{t} \in RV_{-1/\alpha}^{0}$ so that,
again by the representation theorem,
\[
    \beta_{t} = t^{-1/\alpha} \zeta\left( 1/t \right),
\]
as $t \rightarrow 0$, where $\zeta$ is slowly varying at $\infty$. Also, for 
some absolute constant $C>0$ we have
\begin{align*}
    \abs{A_{t} \beta_{t}} &\leq C t \beta_{t} 
    = C t^{1-1/\alpha} \zeta\left( 1/t \right).
\end{align*}
The function $s^{1/\alpha-1} \zeta\left( s \right)$ is regularly
varying at $\infty$ with index $1/\alpha - 1 < 0$. Standard results 
(e.g. Proposition 1.3.6(v) in \cite{bingham}) imply that $s^{1/\alpha - 1} 
\zeta\left( s \right) \rightarrow 0$ as $s \rightarrow \infty$, which shows 
that $A_{t} \beta_{t} \rightarrow 0$ as $ t \rightarrow 0$. So, we need only 
look at the convergence 
\begin{equation}
\beta_{t} X_{t} \Rightarrow Z, \label{eq:conv_without_at}
\end{equation}
as $t \rightarrow 0$ 
under $\Ps$. 

The following result makes precise the relationship between the normalizing
function $\beta_t$ and the slowly varying part of the share-measure tail
$\gamma^{*}$.
\begin{thm}
    \label{thm:beta_behavior}
    Let $x_{0} > 0$, and let $x^{2} \xi_{S}^{*}(x)$ be 
    monotone (increasing or decreasing) for $0< x < x_{0}$. Then 
    \[
        \lim_{t \rightarrow 0} t \beta_{t}^{\alpha} \ell\left(\beta_{t} \right)
        = \Lambda_{\alpha},
     \] 
     where $\Lambda_{\alpha}$ is a positive constant depending only on $\alpha$.
\end{thm}

\begin{rem}
    The constant $\Lambda_{\alpha}$ encapsulates all the slowly varying
    structure of $\ell$; once $\gamma^{*}(x) = x^{-\alpha}\ell(1/x)$ is fixed
    near $0$, the normalization $\beta_t$ is determined up to this factor by
    the relation $t^{-1} \sim \beta_t^{\alpha}\ell(\beta_t)$.  This is the
    key link between the regular variation of the \levy measure and the
    first-order at-the-money asymptotics obtained later.
\end{rem}

\section{First-Order ATM Asymptotics}
\label{sec:atm_asymptotics}

We now state the main pricing results of the paper. We maintain all of
the assumptions and notation from Section~\ref{sec:doa_rv} and add the
following:

\begin{enumerate}[label=(A\arabic*)]
    \item $\nu$ is regularly varying of order $-\alpha$, $\alpha \in (1,2)$, 
        at $0$. \label{assumption:1}
    \item There exist $C>0$ and $x_{1}>0$ (possibly depending on each other) 
    such that $\int_{ \abs{y} \leq x} 
        y^{2} e^{y} \xi\left( y \right) dy \leq C x^{2} \int_{ \abs{y} > x} 
        e^{y} \xi\left( y \right) dy$ for all $x \geq x_{1}$. 
        \label{assumption:2}
    \item There exists $x_{0}>0$ such that $x^{2} \xi_{S}^{*}\left( x \right)$ 
        is monotone
        (either increasing or decreasing) for $0< x < x_{0}$. \label{assumption:3} 
\end{enumerate}

We first make some brief comments on the requirements 
\ref{assumption:1}--\ref{assumption:2}. Assumption~\ref{assumption:1} is
the standard regular--variation condition placing $X$ in the strictly
$\alpha$--stable domain of attraction and determining the small--time scaling
$B_t$. Assumption~\ref{assumption:2} supplies the balance needed to
apply the results of \cite{houdre_marchal} under the share measure.
Their result does not require \ref{assumption:2} in general, but 
it is precisely the
hypothesis ensuring that the truncated second moment of the share-measure
density is negligible relative to its tail, which is the only part of their
inequality used in the proof.

\begin{rem}
The primary technical input is Assumption~\ref{assumption:3}, which imposes a local
regularity condition on the symmetrized share-measure density $\xi_S^*$.
We assume that $x^{2}\xi_S^*(x)$ is eventually monotone as $x\downarrow 0$,
allowing us to apply a monotone density theorem and deduce the key two-sided
estimate~\eqref{eq:xi_equiv} linking $\xi_S^*$ with the tail
$\gamma_S^*(x):=\int_x^\infty \xi_S^*(u)\,du$. Several alternative conditions
yield the same conclusion. The classical monotone density theorem (\cite{bingham}, 
Theorem 1.7.2) shows that if $\gamma_S^*\in\mathrm{RV}_{-\alpha}$ and $\xi_S^*$ is 
ultimately monotone, then $\xi_S^*\in\mathrm{RV}_{-\alpha-1}$ (provided $\alpha>0$). 
More generally, the
O--version of the monotone density theorem (\cite{bingham}, Prop.~2.10.3) replaces
monotonicity with bounded increase/decrease or finite Matuszewska index
assumptions and still yields
\[
\xi_S^*(x) \asymp \frac{\gamma_S^*(x)}{x}, \qquad x\downarrow 0,
\]
which, when $\gamma_S^*\in\mathrm{RV}_{-\alpha}$, implies precisely the
estimate~\eqref{eq:xi_equiv} (where for functions $f$ and $g$, $f \asymp g$ 
if and only if $f=O(g)$ and $g=O(f)$). One may also invoke the smooth-variation
construction (\cite{bingham}, Section 1.8), which shows that any regularly 
varying tail is asymptotically equivalent to a smoothly varying representative 
whose derivative is regularly varying of index $-\alpha-1$. Thus \ref{assumption:3} 
is simply a convenient sufficient condition: the proof requires only an estimate of 
the form \eqref{eq:xi_equiv} for $\xi_S^*$ near the origin, and any of the monotone,
O--regular variation, or smooth-variation assumptions would suffice.
\end{rem}

\begin{thm}
    \label{thm:main_result}
    Along with the conditions \ref{assumption:1}--\ref{assumption:3}, assume that
       \begin{equation}
        \bar{\mu} < \infty,   \label{assumption:4} 
       \end{equation} 
    and let $Z$ be the $\alpha$-stable random variable from \eqref{eq:conv_without_at}.
    Then an ATM European call option has asymptotic expansion
    \begin{align}
        \Expecnb{\left( S_{t} - S_{0} \right)_{+} } = \left( S_{0} \Expecsnb{Z_{+}}
        \right) B_{t} + o\left( B_t  \right), \label{eq:call_price_nobs}
    \end{align}
    as $t \rightarrow 0$.
\end{thm}

\begin{cor}
    \label{cor:no_brownian_case}
    Under the assumptions of Theorem \ref{thm:main_result}, $\hat{\sigma}$,
    the implied volatility of an ATM call option, is such that
    \begin{align}
        \hat{\sigma}\left( t \right) = \sqrt{2 \pi} 
        \frac{B_{t}}{\sqrt{t}} \Expecsnb{Z_{+}} + o\left( \frac{B_{t}}{\sqrt{t}} \right), 
        \label{eq:impliedvol}
    \end{align}
    as $t \rightarrow 0$. 
\end{cor}

We now consider exponential \levy models whose driving process has a nonzero
Gaussian component.  Let $X$ be a \levy process with triplet $(b,\sigma,\nu)$
under the risk–neutral measure, so that $\proc{S_t} = \proc{S_0 e^{X_t}}$ is a martingale.
When $\sigma > 0$ the path decomposition
\[
    X_t \,\eqlaw\, \sigma W_t + L_t,
\]
holds for every fixed $t\ge 0$, where $W$ is a standard Brownian motion and
$L$ is a pure–jump \levy process with triplet $(b,0,\nu)$.  The equality in
law does not impose any independence assumption between $W$ and $L$ and is
used only to separate the diffusive and jump contributions at small times.
The martingale condition $\psi(-i)=0$ fixes the drift $b$ so that
$\E[e^{X_t}]=1$. 

\begin{thm}
\label{thm:Brownian}
Let $\left(X_t\right)_{t\ge 0}$ be a \levy process with triplet $(b,\sigma,\nu)$ under
the risk–neutral measure, where $\sigma>0$.  Then for every $t\ge 0$ the
process admits the distributional decomposition
\[
    X_t \,\eqlaw\, b t + \sigma W_t + L_t,
\]
where $W$ is a standard Brownian motion, $L$ is a pure–jump \levy process
with triplet $(0,0,\nu)$, and $b$ is chosen so that $\E[e^{X_t}]=1$.  
Suppose that $L$ satisfies the hypotheses of
Theorem~\ref{thm:main_result}.  Then the at-the-money call price is such that
\[
    \Expecnb{\left(S_t - S_0\right)_{+}}
        = S_0\, \sigma \sqrt{t}\, \Expecsnb{(W_1)_{+}}
          + o(\sqrt{t}),
    \qquad t\downarrow 0.
\]
\end{thm}

Theorems~\ref{thm:Brownian} and~\ref{thm:BrownianSimplified} show that, 
under the assumptions of Theorem~\ref{thm:main_result}, the jump component is 
always $o(\sqrt{t})$ regardless of its structure, so the leading 
$\sqrt{t}$ behavior of the ATM call price is universal and driven entirely 
by the Gaussian part of the triplet. Theorem~\ref{thm:BrownianSimplified} 
below gives a self-contained formulation of this principle.

\begin{thm}
    \label{thm:BrownianSimplified}
    Let $L = \left( L_{t} \right)_{t \geq 0}$ be a \levy process with triplet $\ltrip$ 
    such that $\Expec{\abs{L_{1}}e^{L_{1}}} < \infty$ and let $S_t = S_0 
    e^{L_{t}}$. Let there exist $B_{t} > 0$ with $B_t \rightarrow 0$ as 
    $t\rightarrow 0$, $\alpha \in (1,2)$, and a probability measure $\Ps$
    such that
    \[
        \frac{1}{B_{t}} \Expecnb{\left( S_{t} - S_0 \right)_{+} } \rightarrow 
        \Expecsnb{Z_{+}},
    \]
    and
    \[
        \Probs{L_{t} \geq B_{t} u } \rightarrow \Probs{Z \geq u}, 
    \]
    for every $u \geq 0$, where $Z$ is an $\alpha$-stable random variable under 
    $\P^{*}$. For $t \geq 0$, let $\proc{X_{t}}$ be given by $X_t = \sigma W_{t} 
    + L_{t}$ where $W = \left( W_{t} \right)_{t \geq 0}$ is a Brownian motion 
    independent of $L$ under $\P^{*}$. If 
    \[
        \frac{B_{t}}{\sqrt{t}} \rightarrow 0,
    \]
    as $t\rightarrow 0$, and if $\proc{S_{t}} = \proc{S_0 e^{X_t}}$ is a martingale
    with respect to its own filtration, then
    \begin{align}
        c(t,0) = \sigma \Expecsnb{(W_{1})_{+}} \sqrt{t} 
        + o\left( \sqrt{t} \right), \label{eq:call_price_bs}
    \end{align}
    as $t\rightarrow 0$. 
\end{thm}

Finally, we obtain the implied volatility asymptotics in a similar way to the case
without Brownian component. 

\begin{cor}
\label{cor:brownian_case}
    Under the hypotheses of Theorem \ref{thm:Brownian}, $\hat{\sigma}$,
    the implied volatility, is such that
    \[
        \hat{\sigma}\left( t \right) = \sigma \sqrt{2 \pi} 
        \Expecsnb{\left( W_{1} \right)_{+}} + o\left( 1 \right),
    \]
    as $t \rightarrow 0$.
\end{cor}

\section{Example: Beyond Power-Law Convergence Rates}
\label{sec:example}

The results of Section~\ref{sec:atm_asymptotics} show that new first-order dynamics 
are possible; to illustrate this, we now present a simple example that satisfies the 
assumptions of Theorem~\ref{thm:main_result} and for which we can compute the rate 
of convergence for the call option price and implied volatility explicitly.

Consider the \levy measure on $\R$ 
\begin{align}
    \nu\left( dx \right) = \begin{cases} 0, & x<-1 \text{ or } x=0 \\ 
        \abs{x}^{-\alpha-1} e^{-x} \abs{\ln{ \abs{x} }} dx, & x \geq -1 \text{ and } 
        x \ne 0,
    \end{cases}
    \label{eq:toy_measure}
\end{align}
with $\alpha \in (1, 2)$, and let $\left(X_{t} 
\right)_{t\geq 0}$ be the \levy process with triplet $\ltrip$, 
with $b$ chosen to satisfy the martingale condition 
\eqref{eq:b_martingale_condition}.

Then,
\begin{equation}
    \gamma^{*}\left( x \right) = \begin{cases}
        \frac{2}{\alpha} x^{-\alpha} \left( \ln{\left(\frac{1}{x} \right)} - 
        \frac{1}{\alpha} \right) + \frac{3}{\alpha^{2}}, & 0<x < 1 \\
            \frac{1}{\alpha^{2} }, & x = 1 \\
            \frac{1}{\alpha} x^{-\alpha} \left( \ln{x} + \frac{1}{\alpha}
            \right), & x > 1,
    \end{cases}
    \label{eq:toy_gamma}
\end{equation}
and so $\gamma^{*}$ is regularly varying of order $-\alpha$ 
at both $0$ and $\infty$. Assumption~\ref{assumption:1} follows directly from \eqref{eq:toy_gamma}, 
  which shows $\gamma^{*} \in RV_{-\alpha}^{0}$. For 
  Assumption~\ref{assumption:2}, the truncated second moment 
  $\int_{\abs{y} \leq x} y^{2} \nu^{*}(dy)$ is bounded uniformly in $x$ by the 
  finiteness of $\Expec{X_{1}^{2} e^{X_{1}}}$, while 
  $x^{2} \gamma^{*}(x) \sim x^{2-\alpha} \ln(1/x)/\alpha$ diverges as 
  $x \to \infty$ (since $2-\alpha > 0$), so the ratio is bounded for all 
  $x \geq x_{1}$ with $x_{1}$ large enough. The requirement \eqref{assumption:4} holds because the odd part of $\xi$ 
near the origin satisfies $\xi(x) - \xi(-x) = x^{-\alpha-1}(e^{-x} - e^{x})
\ln(1/x) = O(x^{-\alpha}\ln(1/x))$, so that 
$\int_{\eta}^{1} x\left(\xi(x) - \xi(-x)\right)dx$ converges as 
$\eta \downarrow 0$ (since $1-\alpha \in (-1,0)$), whence 
$\sup_{\eta} \abs{\mu_{\eta}} < \infty$. Finally, the expression
in \ref{assumption:3} becomes, for $0 < x < 1$,
\[
    2x^{1-\alpha} \ln{\left( \frac{1}{x} \right)},
\]
which is monotone decreasing on $(0,1)$, since differentiation gives 
$x^{-\alpha}\left((1-\alpha)\ln(1/x) - 1\right) < 0$ for all 
$x \in (0,1)$ when $\alpha > 1$.

It is worth noting that verifying \ref{assumption:3} is not strictly necessary here,
since its role in Theorem~\ref{thm:main_result} is to establish 
the form of the \levy density near the origin, which we have directly from 
\eqref{eq:toy_measure}.

All the hypotheses of Theorem~\ref{thm:main_result} are now 
satisfied, and we proceed to determine the behavior of $\beta_{t}$ as 
$t \rightarrow 0$. Again, $B_{t} = 1 / \beta_{t}$ can
be expressed as $t^{1/\alpha} \tilde{\ell}(1/t)$
for $t>0$, where $\tilde{\ell}$ is slowly varying at $\infty$. Further,
there exists $\Lambda>0$ such that
\[
    t \beta_{t}^{\alpha} \ell\left( \beta_{t} \right) \rightarrow \Lambda,
\]
as $t \rightarrow 0$, where $\ell$ is the slowly varying part of $\gamma^{*}$ 
near $0$. 

The functions $\ell$ and $\beta$ are defined up to asymptotic 
equivalence, so we use asymptotic versions that are simpler
to manipulate. From \eqref{eq:toy_gamma}, the slowly varying factor 
satisfies $\ell(y) \sim \frac{2}{\alpha} \ln y$ as $y \to \infty$, so 
Theorem~\ref{thm:beta_behavior} gives
\begin{align*}
    t \beta_{t}^{\alpha} \ln{ \beta_{t} } \rightarrow \Lambda,
\end{align*}
for some $\Lambda > 0$ as $t \to 0$,
which is equivalent to 
\begin{align}
    t \beta_{t}^{\alpha} \ln{ \beta_{t}^{\alpha} } \rightarrow 
    \alpha \Lambda.
    \label{eq:asymp}
\end{align}
Writing $f(x) = x \log{x}$, we rewrite \eqref{eq:asymp} as $t f\left(
\beta_{t}^{\alpha} \right) \rightarrow \alpha \Lambda$ for $t\rightarrow 0$.
Furthermore, we have that $\beta_{t}^{\alpha} \sim f^{-1}\left( 
\alpha \Lambda/t \right)$, as $t\rightarrow 0$ (the function $f$ has an inverse 
for $x$ 
large enough and $\beta_{t}^{\alpha}$ grows large as $t \rightarrow 0$). 
Inverting $f$ requires the Lambert $W$ function: if $f(x) = x\log x$, then 
$f^{-1}(y) = y / W(y)$ where $W$ denotes the principal branch of the 
Lambert function (we use upper-case $W$ here only; it will not be confused 
with Brownian motion in this section since the example has no Gaussian 
component). Since $W(y) \sim \log y$ as $y \to \infty$ (see e.g.\ 
\cite{lambertw}),
\[
    f^{-1}(y) \sim \frac{y}{\log y}, \qquad y \to \infty.
\]

Applying this to $\beta_t^{\alpha} \sim f^{-1}(\alpha\Lambda/t)$ gives
\[
    \beta_{t}^{\alpha} \sim \frac{\alpha \Lambda / t}{\log(\alpha \Lambda / t)},
\]
and therefore
\[
    B_{t} = \frac{1}{\beta_t} = \left( \frac{\alpha \Lambda t}{\log\!\left( 
    \frac{\alpha \Lambda}{t} \right) } \right)^{1/\alpha}.
\]
Ignoring the constants, the first-order rate of convergence is therefore
\[
    \left( \frac{t}{\log{\left( 1/t \right)}} \right)^{\frac{1}{\alpha}}.
\]
Note that we could get even more interesting behavior by introducing further
slowly varying function behavior near the origin (e.g., bounded oscillatory
terms like $\cos(\log(1/x))$).

This example illustrates the basic mechanism for generating new convergence
rates: once $\gamma^*(x)$ is specified near the origin with a slowly varying
modifier, the normalization $B_t$ and hence the first-order call price
asymptotics are determined by solving the relation
$t^{-1} \sim B_t^{-\alpha} \ell(1/B_t)$, which in turn can involve logarithmic,
log--log, oscillatory, or mixed slow variation. By choosing different slowly
varying factors $\ell$ (such as additional iterated logarithms or bounded
oscillatory terms), one can construct \levy models whose at-the-money
convergence rates differ from the classical $t^{1/\alpha}$ behavior in a
controlled way. The de Bruijn conjugate methodology then provides the precise
form of $B_t$ in each case.

\section{Conclusion and Future Work}
\label{sec:conclusion}

In this article we developed a regular–variation framework for the
short-maturity behavior of at-the-money call options in exponential \levy
models and used it to obtain new orders of convergence. Section~\ref{sec:doa_rv} 
established that the stable domain-of-attraction 
structure and the finiteness of the centering constant $\bar{\mu}$ are 
preserved under the share measure, and Section~\ref{sec:atm_asymptotics} 
used these facts to show that, under mild assumptions placing the driving \levy 
process in the domain of attraction of an $\alpha$-stable law with 
$\alpha\in(1,2)$, the at-the-money call price admits the expansion
\[
  \Expecnb{\left(S_t-S_0\right)_+}
    = \bigl(S_0 \Expecsnb{Z_+}\bigr) B_t + o(B_t),
    \qquad t\downarrow 0,
\]
where $B_t$ is the normalizing function from the stable limit and $Z$ is the
limiting stable random variable. All of the distributional input needed for 
this expansion can be read off from the regular variation of the \levy measure 
near zero. Theorem
\ref{thm:beta_behavior} links the small-time scale $B_t$ to the de Bruijn
conjugate of the slowly varying factor in the share-measure tail, and
Theorem~\ref{thm:main_result} then translates this into a general ATM
asymptotic valid for a broad class of pure-jump exponential \levy models.
When a nonzero Brownian component is present, Theorems~\ref{thm:Brownian}
and~\ref{thm:BrownianSimplified} show that the jump contribution is always
lower order, so that the leading $\sqrt{t}$ behavior of the ATM call price is
universal and driven entirely by the Gaussian part of the triplet.

Within this framework we constructed an explicit example for which the
first-order scale $B_t$ deviates from the classical $t^{1/\alpha}$ rate by a
logarithmic correction.  By specifying the share-measure tail
$\gamma^{*}(x)=x^{-\alpha}\ell(1/x)$ with a slowly varying factor $\ell$
involving $\log(1/x)$, we obtained a normalization of the form
\[
  B_t
    \sim \left(\frac{\alpha\Lambda t}{\log(\alpha\Lambda/t)}\right)^{1/\alpha},
\]
and hence an ATM convergence rate driven by $t^{1/\alpha}$ up to a precise
logarithmic penalty.  This example illustrates how stable domains of
attraction combined with regular variation and de Bruijn conjugates can be
used to design models with controlled, nonstandard short-maturity behavior
while still retaining analytic tractability.

\subsection*{Future Work}

Several extensions of the present work appear natural.

\begin{itemize}[leftmargin=1.5em,itemsep=0.25em]
  \item \emph{Higher-order and off-ATM expansions.}
  The domain-of-attraction framework used here is well suited to studying
  higher-order corrections and small-log-moneyness regimes beyond pure
  at-the-money options.  Extending the analysis to joint expansions in time
  and moneyness---for example, in the regime where log-moneyness scales with
  $B_t$---would connect the present results to the full short-maturity smile.

  \item \emph{Relaxing structural assumptions.}
  Assumptions such as the monotonicity of $x^{2}\xi_{S}^{*}(x)$ and
  the specific
  integrability condition \ref{assumption:2} are technically convenient but
  potentially stronger than necessary.  It would be of interest to replace
  these by weaker local regularity or oscillation conditions while still retaining
  control of $B_t$ and the ATM asymptotics, thereby enlarging the admissible
  class of \levy measures.
\end{itemize}

\appendix
\section{Regular Variation}
We give a brief overview of regular variation and direct the reader to
\cite{bingham} for a more comprehensive treatment. First, we introduce the notion of a slowly varying
function.

\begin{defn}
    A nonnegative, measurable function $\ell$ defined on some neighborhood $[M, \infty)$
        of infinity ($M>0$) is \emph{slowly varying at $\infty$}
        if for every $\lambda>0$, we have
        \begin{align}
            \lim_{x \rightarrow \infty} \frac{\ell\left(\lambda x \right)}{
            \ell(x)} = 1.
            \label{eq:sv_def}
        \end{align}
\end{defn}

The next theorem allows us to have a convenient functional form for any slowly varying function. 

\begin{thm}{(Representation Theorem)}
    A function $\ell$ is slowly varying if and only if it can be written in 
    the form 
    \[
        \ell(x) = h(x) \exp{\left( \int_{a}^{x} \varepsilon(u) \frac{du}{u}
    \right)},
    \]
    for $x \geq a$ where $a>0$, $h$ is measurable with $\lim_{x \rightarrow 
    \infty} h(x) = h \in (0,\infty)$, and $\varepsilon$
    is such that $\lim_{x \rightarrow \infty} \varepsilon(x) = 0$. 
\end{thm}

\begin{prop}
    Let $\ell, \ell_{1},$ and $\ell_{2}$ be slowly varying functions. 
    Then:
    \begin{enumerate}[label=(\roman*)]
        \item The function $\ell^{\alpha}$ is slowly varying
            for every $\alpha \in \R$.
        \item The functions $\ell_{1} \ell_{2}$, $\ell_{1} + \ell_{2}$, 
            and (if $\ell_{2}(x) \rightarrow \infty$ as $x \rightarrow \infty$)
            $\ell_{1} \circ \ell_{2}$ are all slowly varying.
        \item For any $\alpha>0$, 
            \begin{align}
                x^{\alpha} \ell(x) \rightarrow \infty \mbox{\;\; and \;\;}
                x^{-\alpha} \ell(x) \rightarrow 0, \label{eq:sv_limits}
            \end{align}
            as $x\rightarrow \infty$. 
    \end{enumerate}
\end{prop}

We are now in a position to introduce the concept of a regularly varying function.

\begin{defn}
    Let $f$ be a positive, measurable function. We say that $f$ is 
    \emph{regularly varying} at $\infty$ if any one of the following 
    conditions holds.
    \begin{enumerate}[label=(\roman*)]
        \item The limit
            \begin{align}
                \lim_{x \rightarrow \infty} \frac{f(\lambda x)}{f(x)} = 
                g(\lambda) \in (0, \infty),
                \label{eq:rv_limit}
            \end{align}
            exists for every $\lambda \in (0,\infty)$. 
        \item The limit \eqref{eq:rv_limit} exists for every $\lambda \in S$
            where $S$ is either a positive-measure subset of $(0,\infty)$
            or a dense subset of $(0,\infty)$. 
        \item The limit \eqref{eq:rv_limit} exists and equals $g(\lambda) = 
            \lambda^{\rho}$ for some $\rho \in \R$. 
        \item The function $f$ has representation
            \begin{align}
                f(x) = x^{\rho} \ell(x),
                \label{eq:rv_formula}
            \end{align}
            where $\rho \in \R$ and $\ell$ is slowly varying at $\infty$. 
    \end{enumerate}
    If any of the above hold, we write $f \in RV_{\rho}^{\infty}$ where 
    $\rho$ is the real found in \eqref{eq:rv_formula}. 
    We use the notation $\ell \in RV_{0}^{\infty}$ if the function $\ell$
    is slowly varying at $\infty$. 
    \end{defn}

The regular variation property of a function defines how $f$ behaves near 
$\infty$. Since only the asymptotic behavior as $x\to\infty$ (or $x\to 0$) matters,
the values of $f$ on any fixed compact interval may be chosen arbitrarily, and we 
assume $f$ is locally bounded on $[0,\infty)$. For our purposes, we will also be 
concerned with functions that are regularly varying at the origin. 
\begin{defn}
    A function $f$ is \emph{regularly varying at $0$} (from the right) with
    index $\rho$ if $f\left( \frac{1}{\cdot} \right) 
    \in RV_{-\rho}^{\infty}$. We denote this by writing $f \in RV_{\rho}^{0}$.
\end{defn}
Combining \eqref{eq:rv_formula} with the representation theorem for slowly
varying functions above gives the standard representation for a regularly varying 
function. That is, $f$ is regularly varying at $\infty$ if and only if it has 
representation
\begin{align}
    f(x) = x^{\rho} h(x) \exp{\left( \int_{a}^{x} \varepsilon(u) 
    \frac{du}{u} \right)}, \label{eq:rep_thm_rv}
\end{align}
for $x \geq a$ where $\rho \in \R$, $h\geq 0$ is a measurable function with a positive, finite
limit at $\infty$, and $\varepsilon$ is a measurable function such that 
$\lim_{x\rightarrow \infty} \varepsilon(x) = 0$.
The following corollary is clear.
\begin{cor}
    Let $f$ be regularly varying with index $\rho \in \R$ at $\infty$. 
    Then,
    \[
        \lim_{x\rightarrow \infty } f(x) = \begin{cases} \infty,& \mbox{if }
            \rho > 0 \\ 0,& \mbox{if } \rho <0. \end{cases}     \]
\end{cor}

There are important results concerning regularly varying functions that we
will need to exploit. These results can be found in \cite{bingham}, and we
list them here without proof.
\begin{thm}{(Potter's Theorem)}
    \label{thm:potter}
    \begin{enumerate}[label=(\roman*)]
        \item If $\ell$ is slowly varying, then for any given constants $A>1$
            and $\delta >0$ there exists $x_{0} = x_{0}(A, \delta)$ such that
            \[
                \frac{\ell(y)}{\ell(x)} \leq A \max{\left\{ \left( \frac{y}{x}
                \right)^{\delta}, \left( \frac{y}{x} \right)^{-\delta}
                \right\} },
            \]
            for all $x,y \geq x_{0}$. 
        \item If $f \in RV_{\rho}^{\infty}$, then for any given $A>1$ and 
            $\delta>0$ there exists $x_{0} = x_{0}(A,\delta)$ such that
            \[    
                \frac{f(y)}{f(x)} \leq A \max{\left\{ \left( \frac{y}{x}
                \right)^{\rho + \delta}, \left( \frac{y}{x} \right)^{\rho -
                \delta} \right\} },
            \]
            for all $x,y \geq x_{0}$. 
    \end{enumerate}
\end{thm}

\begin{thm}{(Karamata's Theorem)}
    \label{thm:karamata}
    \begin{enumerate}[label=(\roman*)]
        \item (Direct Half) If $\ell$ is slowly varying, $\ell$ is locally 
        bounded on $[x_{0}, \infty)$ for some $x_0>0$, and $\alpha >-1$, 
            then
            \[
                \lim_{x \rightarrow \infty} \frac{ \int_{x_{0}}^{x} u^{\alpha} 
                \ell(u) du}{ x^{\alpha +1} \ell(x)} = \frac{1}{(\alpha + 1)}.
            \]
        \item (Converse Half) Let $f$ be positive and locally integrable
            on $[x_{0}, \infty)$.
                \begin{itemize}
                    \item If for some $\zeta > - (\rho + 1)$, 
                        \[
                            \lim_{x \rightarrow \infty} \frac{x^{\zeta + 1} 
                        f(x)}{ \int_{x_{0}}^{x} 
                            u^{\zeta} f(u) du} = \zeta + \rho + 1,
                        \]
                        then $f$ varies regularly with index $\rho$.
                    \item If for some $\zeta < - \left( \rho + 1 \right)$
                        we have
                        \[
                            \lim_{x\rightarrow \infty} \frac{x^{\zeta + 1} f(x)}{ \int_{x}^{\infty} 
                            u^{\zeta} f(u) du} = -\left( \zeta +
                            \rho + 1\right),
                        \]
                        then $f$ varies regularly with index $\rho$. 
                \end{itemize}
    \end{enumerate}
\end{thm}

\begin{thm}{(de Bruijn Conjugate)}
    If $f \in RV_{\alpha}^{\infty}$ with $\alpha>0$, then there exists 
    $g \in RV_{1/\alpha}^{\infty}$ such that
    \[
        \lim_{x \rightarrow \infty} \frac{f\left( g\left( x \right) \right)}{x} = 
        \lim_{x \rightarrow \infty} \frac{g\left( f\left( x \right) \right)}{x} = 1,
    \]
    i.e. $f$ and $g$ are asymptotically invertible.
\end{thm}

A natural question is to know whether or not the derivative of a regularly varying function 
is itself regularly varying. The following result shows that monotonicity of 
the derivative in a neighborhood of $\infty$ suffices.

\begin{thm}{(Monotone Density Theorem)}
    \label{thm:monotone_density}
    Let $H(x) = \int_{0}^{x} h(u)\,du$ where $h:[0,\infty) \to \R$ is measurable. 
    If $H(x) \sim c x^{\rho} \ell(x)$ as $x \rightarrow \infty$ where $c, 
    \rho \in \R$ and $\ell \in 
    RV_{0}^{\infty}$ and if there exists $x_{0}>0$ such that 
    $h(x)$ is monotone on $(x_{0}, \infty)$, then 
    \[
        \lim_{x \rightarrow \infty} \frac{ h(x)}{ c \rho x^{\rho-1} \ell(x) } = 1.
    \]
\end{thm}

It is important to note that Theorem \ref{thm:monotone_density} does 
\emph{not} imply that $H$ or $h$ is regularly varying, as the quantities $c$ 
and $c \rho$ are potentially negative; however, if $c>0$, then $H \in 
RV_{\rho}^{\infty}$, but we still do not necessarily have $h \in 
RV_{\rho - 1}^{\infty}$.

Additionally, though we use Theorem \ref{thm:monotone_density}, there are ways to 
weaken the hypotheses of our main result to not require eventual monotonicity. For this we 
refer the reader to the O-version of Theorem \ref{thm:monotone_density} (which requires 
bounded increase and bounded decrease assumptions) in \cite{bingham} Chapter 2.10  
(e.g., Prop 2.10.3). 

Most of the corresponding definitions and theorems extend straightforwardly 
to regular variation at $0$ by considering $f(1/x)$ as $x\to \infty$, and we omit
these routine adaptations.
\section{Proofs}

\begin{proof}[Proof of Proposition \ref{lem:preserve_reg_var}]
    Let $\nu$ be regularly varying of index $\alpha$ at $0$, i.e.,
    \[
        \lim_{r \rightarrow 0} \frac{\nu\left( \abs{x} >rt \right)}{\nu\left( 
            \abs{x}>r \right)} = t^{-\alpha},
    \]
    for all $t>0$, or equivalently
    \[
        \lim_{r \rightarrow 0} \frac{\int_{rt}^{\infty} \xi_{S}(x)dx}{ 
            \int_{r}^{\infty} \xi_{S}(x)dx} = t^{-\alpha}.
    \]
    First, we show that for any fixed $\delta >0$ and $t>0$ we also have
    \[
        \lim_{r \rightarrow 0} \frac{\int_{rt}^{\delta} \xi_{S}(x)dx}{
            \int_{r}^{\delta} \xi_{S}(x)dx} = t^{-\alpha}.         
    \]
    Recall that $\gamma\left( x \right) = \int_{x}^{\infty} \xi_{S}(z) dz 
    \rightarrow \infty$ as $x \rightarrow 0$ by the representation theorem
    for regularly varying functions (see \cite{bingham}). 
    So for fixed $\delta>0$, 
    \begin{align}
        \lim_{r \rightarrow 0} \frac{\int_{rt}^{\delta} \xi_{S}(x) dx}{
            \int_{r}^{\delta} \xi_{S}(x)dx} 
        &= \lim_{r \rightarrow 0} \frac{\int_{rt}^{\infty} \xi_{S}(x) dx - 
        \int_{\delta}^{\infty} \xi_{S}(x) dx}{
            \int_{r}^{\infty} \xi_{S}(x) dx - \int_{\delta}^{\infty}
            \xi_{S}(x) dx} 
        = \lim_{r \rightarrow 0} \frac{\int_{rt}^{\infty} \xi_{S}(x) dx 
        \left( 1 - \frac{\int_{\delta}^{\infty} \xi_{S}(x) dx}{\int_{rt}^{
        \infty} \xi_{S}(x) dx} \right)}{ \int_{r}^{\infty} \xi_{S}(x) dx 
        \left(1 - \frac{\int_{\delta}^{\infty} \xi_{S}(x) dx}{\int_{r}^{\infty} 
        \xi_{S}(x) dx} \right)} 
        = t^{-\alpha}.
    \end{align}
    Continuing, we let $0 < \varepsilon <1 $ and choose $\delta >0$ such that 
    \[
        1-\varepsilon \leq e^{x} \leq 1+\varepsilon,
    \]
    for all $x \in (-\delta,\delta)$. Now, recalling 
    \eqref{eq:levy_density_symmetrization_star}, we have
    \begin{align}
        \frac{\nu^{*}\left( \abs{x}> rt \right)}{\nu^{*}\left( \abs{x}>r 
    \right)} &= \frac{\int_{\abs{x}>rt}e^{x} \xi(x) dx}{ \int_{\abs{x}>r} 
    e^{x} \xi(x)dx} 
        = \frac{\int_{rt}^{\infty}\left( e^{x} \xi(x) + e^{-x} \xi(-x) 
    \right)dx}{ \int_{r}^{\infty} \left( e^{x} \xi(x)  + e^{-x} \xi(-x) 
\right) dx}
    = \frac{\int_{rt}^{\delta} \xi^{*}_{S}(x) dx + C_{\delta}^{*}}{ 
        \int_{r}^{\delta} \xi^{*}_{S}(x) dx + C_{\delta}^{*}},
            \label{eq:fraction_int}
    \end{align}
    where $C_{\delta}^{*} = \int_{\delta}^{\infty} \xi^{*}_{S}(x) dx < \infty$. 
    For $0<x<\delta$, define $\gdelsof{x} =  \int_{x}^{\delta} 
    \xi^{*}_{S}(y) dy$ and $\gdelof{x} = \int_{x}^{\delta} \xi_{S}(y) dy$ 
    as the truncated tail functions. Note that $\gdelof{x} \rightarrow 
    \infty$ as $x \rightarrow 0$ by the representation theorem
    for regularly varying functions. We estimate 
    \eqref{eq:fraction_int} from above as
    \begin{align}
        \frac{\int_{rt}^{\delta} \xi^{*}_{S}(x) dx + C_{\delta}^{*}}{ 
            \int_{r}^{\delta} \xi^{*}_{S}(x) dx + C_{\delta}^{*}}
        &\leq \frac{\int_{rt}^{\delta} e^{x} \xi_{S}(x) dx + 
        C_{\delta}^{*}}{\int_{r}^{\delta} e^{-x} \xi_{S}(x) dx +
        C_{\delta}^{*}} \notag \\
        &\leq \frac{(1 + \varepsilon) \int_{rt}^{\delta} \xi_{S}(x) dx + 
        C_{\delta}^{*}}{ (1-\varepsilon) \int_{r}^{\delta} \xi_{S}(x) dx + 
        C_{\delta}^{*}}
        = \frac{\gdelof{rt}}{\gdelof{r}} \left( \frac{1+\varepsilon +
                \frac{C_{\delta}^{*}}{\gdelof{rt}}}{1-\varepsilon + 
                    \frac{C_{\delta}^{*}}{\gdelof{r}}} \right),
                    \label{eq:upper_bound_rv}
    \end{align}
    and from \eqref{eq:upper_bound_rv} obtain
    \[
        \limsup_{r \rightarrow 0} \frac{\nu^{*}\left( \abs{x}> rt \right)}{
            \nu^{*}\left(\abs{x}>r \right)} \leq t^{-\alpha} \frac{1+
            \varepsilon}{1-\varepsilon}.
    \]
    We estimate \eqref{eq:fraction_int} from below 
    \begin{align}
        \frac{\int_{rt}^{\delta} \xi^{*}_{S}(x) dx + C_{\delta}^{*}}{ 
            \int_{r}^{\delta} \xi^{*}_{S}(x) dx + C_{\delta}^{*}}
        &\geq \frac{\int_{rt}^{\delta} e^{-x} \xi_{S}(x) dx + C_{\delta}^{*}}{
            \int_{r}^{\delta} e^x \xi_{S}(x) dx + C_{\delta}^{*}}
        \geq \frac{\gdelof{rt}}{\gdelof{r}} \left( \frac{1-\varepsilon +
                \frac{C_{\delta}^{*}}{\gdelof{rt}}}{1+\varepsilon + 
                    \frac{C_{\delta}^{*}}{\gdelof{r}}} \right), 
                    \label{eq:lower_bound_rv}
    \end{align}
    and from \eqref{eq:lower_bound_rv} obtain
    \[
        \liminf_{r \rightarrow 0}  \frac{\nu^{*}\left( \abs{x}> rt \right)}{
            \nu^{*}\left( \abs{x}>r \right)} \geq t^{-\alpha} \frac{1-
            \varepsilon}{1+\varepsilon}. 
    \]
    Thus, 
    \[
        t^{-\alpha} \frac{1-\varepsilon}{1+\varepsilon} \leq \liminf_{r 
            \rightarrow 0} \frac{\nu^{*}\left( \abs{x}> rt \right)}{
                \nu^{*}\left( \abs{x}>r 
        \right)} \leq \limsup_{r \rightarrow 0} \frac{\nu^{*}\left( \abs{x}> rt 
        \right)}{\nu^{*}\left( \abs{x}>r \right)} \leq t^{-\alpha} \frac{1+
        \varepsilon}{1-\varepsilon},
    \]
    and letting $\varepsilon \rightarrow 0$ gives the result. 

    Next, we need to show the existence of the limits 
    \[
        \lim_{x \rightarrow 0} \frac{\gamma^{*}_{\pm} (x)}{ 
            \gamma^{*}(x)},
    \]
    given the existence of the limits
    \[
        \lim_{x \rightarrow 0} \frac{\gamma_{\pm} (x)}{ 
            \gamma(x)}.
    \]
    Assume that $\lim_{x \rightarrow 0} \gamma_{+}(x) / \gamma(x) = p_+$ 
    and $\lim_{x \rightarrow 0} \gamma_{-}(x) / \gamma(x) = p_-$. 
    In fact, we show that the limits above are equal, i.e.
    \begin{align}
        \lim_{x \rightarrow 0} \frac{\gamma^{*}_{+} (x)}{ 
            \gamma^{*}(x)} = \lim_{x \rightarrow 0} 
            \frac{\gamma_{+} (x)}{ 
            \gamma(x)} = p_+,
        \label{eq:p_lim_equality}
    \end{align}
    and
    \begin{align}
        \lim_{x \rightarrow 0} \frac{\gamma^{*}_{-} (x)}{ 
            \gamma^{*}(x)} = \lim_{x \rightarrow 0} 
            \frac{\gamma_{-} (x)}{ 
            \gamma(x)} = p_-.
        \label{eq:q_lim_equality}
    \end{align}
    By an argument similar to the one developed in the 
    beginning of this proof, the limits can also be computed via
    \begin{align*}
        \lim_{x \rightarrow 0} &\frac{\nu\left( x< y < \delta \right)}{
            \nu\left( x < \abs{y} < \delta \right)} = p_+, 
    \end{align*}
    and
    \begin{align*}
        \lim_{x \rightarrow 0} &\frac{\nu\left( -\delta < y < -x \right)}{
            \nu\left( x < \abs{y} < \delta \right)} = p_-,
    \end{align*}
    where $\delta>0$. We now show that the same limits hold for 
    $\gamma^{*}$. 
    
    First, let $0<p_+<1$ (note that $p_+ + p_-=1$). 
    In this case, $0<p_-<1$, and so 
    \[
        \lim_{x \rightarrow 0} \gamma_{+}^{*}(x) = 
        \lim_{x \rightarrow 0} \gamma_{-}^{*}(x) = \infty,
    \]
    since $\lim_{x \rightarrow 0} \gamma^{*}(x) = \infty$. Again, let 
    $\varepsilon>0$ and choose $\delta>0$ such that $1-\varepsilon \leq e^{x} 
    \leq 1+\varepsilon$ for all $x\in (-\delta,\delta)$. Continuing     
    \begin{align}
        \frac{\nu^{*}\left( y > x \right)}{\nu^{*}\left( \abs{y}>x \right)}
            &= \frac{\int_{x}^{\infty} \xi^{*}(y) dy}{ 
                \int_{x}^{\infty} \xi_{S}^{*}(y) dy}
            = \frac{\int_{x}^{\delta} e^{y} \xi(y) dy + 
                \int_{\delta}^{\infty} e^{y} \xi(y) dy}{
                \int_{x}^{\delta} \xi_{S}^{*}(y) dy + 
                \int_{\delta}^{\infty} \xi_{S}^{*}(y) dy} 
            = \frac{\int_{x}^{\delta} e^{y} \xi(y) dy + 
                D_{\delta}^{*}}{ \int_{x}^{\delta} \xi_{S}^{*}(y) dy + 
                C_{\delta}^{*}}, \label{eq:p_mass_ratio} 
    \end{align}
    where $D_{\delta}^{*} = \int_{\delta}^{\infty} e^{y} \xi(y) dy$ and, again,
    $C_{\delta}^{*} = \int_{\delta}^{\infty} \xi_{S}^{*}(y) dy$ are constants 
    depending only on $\delta$. Estimating \eqref{eq:p_mass_ratio} above 
    by
    \begin{align}
        \frac{\int_{x}^{\delta} e^{y} \xi(y) dy + 
                D_{\delta}^{*}}{ \int_{x}^{\delta} \xi_{S}^{*}(y) dy + 
                C_{\delta}^{*}} 
            &\leq \frac{\int_{x}^{\delta} e^{y} \xi(y) dy + 
            D_{\delta}^{*}}{ \int_{x}^{\delta} e^{-y} \xi_{S}(y) dy + 
            C_{\delta}^{*}} \notag \\
            &\leq \frac{(1+ \varepsilon) \int_{x}^{\delta} \xi(y) dy + 
            D_{\delta}^{*}}{ (1-\varepsilon) \int_{x}^{\delta} \xi_{S}(y) dy + 
            C_{\delta}^{*}}
            = \frac{\int_{x}^{\delta} \xi(y) dy}{
                \int_{x}^{\delta} \xi_{S}(y) dy} 
                \left( \frac{1 + \varepsilon + \frac{D_{\delta}^{*}}{
                    \int_{x}^{\delta} \xi(y)dy }}{ 1 - \varepsilon +
                        \frac{C_{\delta}^{*}}{\int_{x}^{\delta} \xi_{S}(y) dy}} 
                        \right), \label{eq:p_mass_limsup}
    \end{align}
    which, taking the limsup, gives
    \[
        \limsup_{x \rightarrow 0} \frac{\nu^{*}(y>x)}{\nu^{*}(\abs{y}>x)} \leq
        p_+ \left( \frac{1+\varepsilon}{1-\varepsilon} \right),
    \]
    (since $\lim_{x\rightarrow 0} \int_{x}^{\delta} \xi(y) dy = \lim_{x 
    \rightarrow 0} \int_{x}^{\delta} \xi_{S}(y) dy = \infty$ also for every 
    $\delta>x$). We estimate \eqref{eq:p_mass_ratio} from below as
    \begin{align}
    \frac{\int_{x}^{\delta} e^{y} \xi(y) dy + 
            D_{\delta}^{*}}{ \int_{x}^{\delta} \xi_{S}^{*}(y) dy + 
            C_{\delta}^{*}} 
        &\geq \frac{\int_{x}^{\delta} e^{y} \xi(y) dy + 
        D_{\delta}^{*}}{ \int_{x}^{\delta} e^{y} \xi_{S}(y) dy + 
        C_{\delta}^{*}} \notag \\
        &\geq \frac{ \int_{x}^{\delta} \xi(y) dy + 
        D_{\delta}^{*}}{ (1+\varepsilon) \int_{x}^{\delta} \xi_{S}(y) dy + 
        C_{\delta}^{*}} 
        = \frac{\int_{x}^{\delta} \xi(y) dy}{
            \int_{x}^{\delta} \xi_{S}(y) dy} 
            \left( \frac{1 + \frac{D_{\delta}^{*}}{
                \int_{x}^{\delta} \xi(y)dy }}{ 1 + \varepsilon +
                    \frac{C_{\delta}^{*}}{\int_{x}^{\delta} \xi_{S}(y) dy}} \right),
                    \label{eq:p_mass_liminf}
    \end{align}
    and taking the liminf gives
    \[
        \liminf_{x \rightarrow 0} \frac{\nu^{*}(y>x)}{\nu^{*}(\abs{y}>x)} \geq
        p_+ \left( \frac{1}{1+\varepsilon} \right).
    \]
    Combining these estimates, we obtain
    \[
        p_+ \left( \frac{1}{1+\varepsilon} \right) \leq
        \liminf_{x \rightarrow 0} \frac{\nu^{*}(y>x)}{\nu^{*}(\abs{y}>x)} \leq
        \limsup_{x \rightarrow 0} \frac{\nu^{*}(y>x)}{\nu^{*}(\abs{y}>x)} \leq
        p_+ \left( \frac{1+\varepsilon}{1-\varepsilon} \right),
    \]
    and letting $\varepsilon \rightarrow 0$ gives the first limit. 
    An identical argument shows that 
    \[
        \lim_{x \rightarrow 0} \frac{\nu^{*}(y<-x)}{\nu^{*}(\abs{y}>x)} = p_-.
    \]

    We now deal with the remaining cases, i.e. $p_+=0$ and $p_+=1$, and assume without
    loss of generality that $p_+=0$. This implies that
    \begin{align*}
        \lim_{x \rightarrow 0} &\frac{\gamma_{+}(x)}{\gamma(x)} = 0,
    \end{align*}
    and
    \begin{align*}
        \lim_{x \rightarrow 0} &\frac{\gamma_{-}(x)}{ \gamma(x)} = 1,
    \end{align*}
    which in turn implies that $\lim_{x \rightarrow 0} \gamma_{-}(x) = \infty$. 
    There are two distinct possibilities for $\gamma_{+}$, either
    \begin{align}
        \lim_{x \rightarrow 0} \gamma_{+}(x) < \infty,
        \label{eq:const_case} 
    \end{align}
    or
    \begin{align}
        \lim_{x \rightarrow 0} \gamma_{+}(x) &= \infty .
        \label{eq:infinite_case}
    \end{align}
    If \eqref{eq:infinite_case} holds true, then both tails have infinite mass
    and a proof similar to the one for $0<p_+<1$ gives the result. In the case
    \eqref{eq:const_case},
    \begin{align*}
        \lim_{x\rightarrow 0} \frac{\gamma_{+}^{*}(x)}{\gamma^{*}(x)} = 0,
    \end{align*}
    since $\gamma^{*}(x) \rightarrow \infty$ as $x \rightarrow 0$. Indeed, we
    estimate
    \begin{align*}
        \gamma^{*}(x) &= \int_{x}^{\infty} \xi_{S}^{*}(y) dy
        = \int_{x}^{\delta} \xi_{S}^{*}(y) dy + C_{\delta}^{*} 
        \geq \int_{x}^{\delta} e^{-y} \xi_{S}(y) dy + C_{\delta}^{*}
        \geq (1-\varepsilon) \int_{x}^{\delta} \xi_{S}(y) dy + C_{\delta}^{*} 
        \rightarrow \infty,
    \end{align*}
    as $x \rightarrow 0$. Thus, 
    \begin{align*}
        \lim_{x \rightarrow 0} \frac{\gamma_{-}^{*}(x)}{\gamma^{*}(x)} =
        \lim_{x \rightarrow 0} \frac{\gamma^{*}(x) - \gamma_{+}^{*}(x)}{
            \gamma^{*}(x)} = 1.
    \end{align*}
\end{proof}

\begin{proof}[Proof of Proposition \ref{lem:mu_bar_finite}]
    First, we assume $\bar{\mu}_{0} < \infty$. Observe that 
        $\sup_{0<\eta \leq 1} \abs{\mu_{\eta}} < \infty$ implies that 
    \[
        \sup \mu_{\eta} < \infty \mbox{\;\;\; and \;\;\;} \inf \mu_{\eta} 
        > -\infty. 
    \]
    Thus, 
    \begin{align*}
        -\infty &< \inf_{0< \eta \leq 1} \left( b - \int_{\eta < \abs{y} 
    \leq 1}  y \nu(dy) \right)
            = b + \inf_{0< \eta \leq 1} \left( - \int_{\eta < \abs{y} \leq 1} 
            y \nu(dy) \right)
            = b - \sup_{0< \eta \leq 1}  \int_{\eta < \abs{y} \leq 1} y 
            \nu(dy), 
    \end{align*}
    which implies $\sup_{0< \eta \leq 1} \int_{\eta < \abs{y} \leq 1} y 
    \nu(dy) < \infty$. A similar argument implies that 
    \[
        \inf_{0< \eta \leq 1} \int_{\eta < \abs{y} \leq 1} y \nu(dy) > -\infty,
    \]
    and so we know
    \[
        \sup_{0< \eta \leq 1} \abs{\int_{\eta < \abs{y} \leq 1} y \nu(dy)} < 
        \infty.
    \]
    For fixed $0< \eta \leq 1$, 
    \begin{align*}
        \abs{\mu^{*}_{\eta} } &\leq \abs{b} + \abs{ \int_{\eta < \abs{y} \leq 
    1} y \nu^{*} (dy) } \\
        &= \abs{b} + \abs{ \int_{\eta < \abs{y} \leq 1} y e^{y} \nu(dy) }
        \leq \abs{b} + e \abs{ \int_{\eta < \abs{y} \leq 1} y \nu(dy) }
        \leq \abs{b} + e \sup_{0< \eta \leq 1} \abs{ \int_{\eta < \abs{y} 
    \leq 1} y \nu(dy) } < \infty . 
    \end{align*}
    Taking the supremum gives the implication.

    The converse can be proven by noting that
    \begin{align*}
        \abs{\int_{ \eta < \abs{y} \leq 1} y e^{y} \nu\left( dy \right) }
                &\geq \frac{1}{e} \abs{\int_{ \eta < \abs{y} \leq 1} y \nu\left( 
                dy \right) },
    \end{align*}
    and taking the supremum since the left-hand side is bounded when 
    $\eta \rightarrow 0$ by our assumption. 
\end{proof}

\begin{proof}[Proof of Theorem \ref{thm:beta_behavior}]
    We know
    \begin{align}
        \beta_{t} X_{t} \Rightarrow Z, 
    \end{align}
    as $t \rightarrow 0$, under $\Ps$ where $Z$ is an $\alpha$-stable random 
    variable. Recall that~\cite{grabchak} and~\cite{bingham} both give the 
    representation 
    \begin{align}
        \varphi_{Z}(u) = \exp{\bigl( -c_{\alpha} \abs{u}^{\alpha} 
        \left(1 - i(p_{+}-p_{-})
        \sgn{(u) \tan{(\pi \alpha/2)}} \right) \bigr)}, \label{eq:convergence} 
    \end{align}
    where $p_{\pm}$ are defined in \eqref{eq:tail_limits}.
    We need further information concerning the slowly varying
    part of $\gamma^{*}$. So, we examine the characteristic functions of both 
    $\beta_{t} X_{t}$ and $Z$, which we know must be equal when $t \rightarrow 0$
    by \eqref{eq:conv_without_at}. 
    
    First, we will need to determine the behavior of $\xi^{*}_{S}(x)$. 
    To this end, we will use the Monotone Density Theorem (Theorem 
    \ref{thm:monotone_density}). 
    We know $\gamma^{*}(x) = \int_{x}^{\infty} \xi^{*}_{S}(y) dy = x^{-\alpha} 
    \ell(1/x)$, for $x>0$. Hence,
    \begin{align}
        x^{-\alpha} \ell(x) &= \gamma^{*}\left( 1/x \right)
        = \int_{1/x}^{\infty} \xi^{*}_{S}(y) dy
        = - \int_{x}^{0} \xi^{*}_{S}\left( 1/u \right)\frac{du}{u^{2}} 
        = \int_{0}^{x} \xi^{*}_{S} \left( 1/u \right) \frac{du}{u^{2}} 
        = \int_{0}^{x} s(u) du, \label{eq:density_trans}
    \end{align}
    where $s(u):= \xi_{S}^{*}\left( 1/u \right) u^{-2}$.
    Note that $x^{2} \xi^{*}_{S}(x) = s(1/x)$ and $x^2 \xi^{*}_{S}$ 
    is monotone for $x$ close enough to $0$ (i.e. $s(u)$ is monotone
    for $u$ large enough). Now, we use the Monotone Density Theorem to get 
    that $s(x) \sim \alpha x^{\alpha-1} \ell(x)$, as $x\rightarrow \infty$. 
    That is, for $y$ close to $0$, we have $s(1/y) = y^{2} \xi^{*}_{S}(y) \sim 
    \alpha y^{-\alpha+1} \ell(1/y)$ which implies $\xi^{*}_{S}(y) 
    \sim \alpha y^{-\alpha-1} \ell(1/y)$ for $y$ positive and near $0$. 
    Now, the exponent of the characteristic function of $\beta_{t} X_{t}$ is
    given by
    \begin{align}
        \log{\left(\Expecsnb{e^{iu\beta_{t} X_{t}}}\right)} &= t 
        \int_{-\infty}^{\infty} \left(  \exp{(iu\beta_{t}y)} - 1 - i u 
        \beta_{t} y\right) \xi^{*}(y) dy.
        \label{eq:integral} 
    \end{align}
    Fix any $0<\varepsilon<1$ and let $w_0 (\varepsilon) >0$ be such that 
    \begin{align}
        (1-\varepsilon) \alpha x^{-\alpha-1} \ell(1/x) \leq \xi^{*}_{S}(x) 
        \leq (1+ \varepsilon) \alpha x^{-\alpha-1} \ell(1/x),
        \label{eq:xi_equiv}
    \end{align}
    for all $0<x<w_0$. The real part of \eqref{eq:integral} 
    converges to $-c_{\alpha} \abs{u}^{\alpha}$ where $c_{\alpha}>0$ as $t \rightarrow 0$
    (again see~\cite{grabchak},~\cite{bingham}, and~\cite{maller_mason}).
    First, we need to rewrite \eqref{eq:integral} in a nicer form. Let 
    $g(u,y) = \exp{(iu\beta_{t}y)} - 1 - i u \beta_{t} y$ and rewrite
    \begin{align}
        \log{\left(\Expecsnb{e^{iu\beta_{t} X_{t}}}\right)} &= t 
        \int_{-\infty}^{\infty} g(u,y) \xi^{*}(y) dy
        = t  \int_{0}^{\infty} \left( g(u,y) \xi^{*}(y) + 
        \overline{g(u,y)} \xi^{*}(-y) 
        \right)dy.
        \label{eq:xi_int}
    \end{align}
    Note that the real part of $g$ is $\Re{(g(u,y))} = \cos{(u \beta_{t} y)} - 1$, 
    and the real part of \eqref{eq:xi_int} is
    \begin{align}
        \Re{\left( \log{\left(\Expecsnb{e^{iu\beta_{t} X_{t}}} \right)} \right) }
        &= t \int_{0}^{\infty}  \Re{g(u,y)} \left( \xi^{*}(y) + \xi^{*}(-y) 
        \right) dy \notag \\
        &= t \int_{0}^{\infty} \left( \cos{(u \beta_{t} y)} - 1 \right) 
        \xi^{*}_{S}(y) dy \notag \\
        &= \frac{t}{\abs{u} \beta_{t}} \int_{0}^{\infty} \left( 
        \cos{(\sgn{(u)}w)} - 1 \right) \xi^{*}_{S}\left( 
        \frac{w}{\abs{u} \beta_{t}} \right) dw \notag \\
        &= \frac{t}{\abs{u} \beta_{t}} \int_{0}^{\infty} \left( \cos{(w)} - 1 
        \right) \xi^{*}_{S}\left( \frac{w}{\abs{u} \beta_{t}} \right) dw.
        \label{eq:real_part_int}
    \end{align}
    Continuing, we break \eqref{eq:real_part_int} into two parts by writing 
    for $L>0$
    \begin{align}
        \Re{\left( \log{\left(\Expecsnb{e^{iu\beta_{t} X_{t}}}\right)} \right) } 
        &= \frac{t}{\abs{u} \beta_{t}} \int_{0}^{\infty} \left( \cos{(w)} - 1 
        \right) \xi^{*}_{S}\left( \frac{w}{\abs{u} \beta_{t}} \right) 
        dw  \notag \\
        &= \frac{t}{\abs{u} \beta_{t}} \int_{0}^{\infty} \left( \cos{(w)} - 1 
        \right) \xi^{*}_{S} \left( \frac{w}{\abs{u} \beta_{t}} \right) 
        \Ind{\frac{w}{\abs{u} \beta_{t}} \leq L} dw \label{eq:ch_fn_p1} \\
        &\;\;\;\; + \frac{t}{\abs{u} \beta_{t}} \int_{0}^{\infty} \left( \cos{(w)} - 1 
        \right) \xi^{*}_{S} \left( \frac{w}{\abs{u} \beta_{t}} \right) 
        \Ind{\frac{w}{\abs{u} \beta_{t}} > L} dw. \label{eq:ch_fn_p2} 
    \end{align}
    It is easy to see that \eqref{eq:ch_fn_p2} goes to $0$, as $t \rightarrow 0$, 
    since
    \begin{align*}
        \frac{t}{\abs{u} \beta_{t}} &\abs{\int_{0}^{\infty}  \left( \cos{(w)} 
        - 1 \right) \xi^{*}_{S} \left( \frac{w}{\abs{u} \beta_{t}} \right)  
        \Ind{\frac{w}{\abs{u} \beta_{t}} > L} dw} \\
        &= t \abs{\int_{0}^{\infty} \left( \cos{(u\beta_{t} z)} - 1 \right) 
        \xi^{*}_{S}(z) \Ind{z \geq L} dz }
        \leq 2t \int_{L}^{\infty} \xi^{*}_{S}(z) dz.
    \end{align*}
    Now, we estimate \eqref{eq:ch_fn_p1} to get the desired result. 
    First, we show a preliminary result of use later. Namely, we 
    show that there exists $M>0$ with $1/M \leq w_0$ such that
    \begin{align}
        \int_{0}^{\infty} \frac{\left( \cos{(w)} -1 \right)}{w^{1+\alpha}} 
        \frac{\ell\left( \frac{\abs{u} \beta_{t}}{w} \right)}{\ell\left( 
            \beta_{t} \right)} \Ind{\frac{w}{\abs{u} \beta_{t}} \leq 
            \frac{1}{M}} dw \rightarrow \int_{0}^{\infty} 
        \frac{\left( \cos{(w)} -1 \right)}{w^{1+\alpha}}dw < \infty, 
        \label{eq:int_w_SV}
    \end{align}
    as $t \rightarrow 0$. Recall that $\ell$ slowly varying implies that 
    $\ell(\lambda x)/\ell(x) \rightarrow 1$, for any $\lambda >0$, as 
    $x\rightarrow \infty$. Letting $x=\beta_{t}$ and $\lambda = \abs{u}/w$ 
    implies that 
    \[
        \frac{\ell\left( \frac{\abs{u} \beta_{t}}{w} \right)}{\ell(\beta_{t})} 
        \rightarrow 1,
    \]
    as $t \rightarrow 0$. We recall the Potter bounds from Theorem 
    \ref{thm:potter} for $\ell$, that is for any 
    $A>1$ and $\delta>0$ there exists $M>0$ such that for $x,y \geq M$,
    \begin{align*}
        \frac{\ell(y)}{\ell(x)} \leq A \left( \left( \frac{y}{x} \right)^{
        \delta} \vee \left( \frac{y}{x} \right)^{-\delta}  \right) .
    \end{align*}
    Choose $A=2$ and $\delta >0$ such that $\alpha+1 \pm \delta \in (2,3)$ and 
    let $M_0 >0$ be the $M$ in the above statement for the given $A$ and 
    $\delta$. Note that if $w \leq \abs{u} \beta_{t} w_0$ and $w_{0} M_0 \leq 1$, 
    then $\abs{u}\beta_t / w \geq M_0$, so the Potter bounds apply.
    Since the integrand is non-positive, we bound its absolute value 
    using the Potter bounds. We have
    \begin{align}
        (1 - \cos{w}) w^{-\alpha-1} \frac{\ell\left( \frac{\abs{u} 
        \beta_{t}}{w} \right)}{\ell(\beta_{t})} 
        \Ind{\frac{w}{\abs{u} \beta_{t}}\leq w_0} 
        &\leq (1 - \cos{w}) w^{-\alpha-1} 
        A \max \left\{ \left( \frac{\abs{u}}{w} \right)^{\delta}, \left( 
            \frac{\abs{u}}{w} \right)^{-\delta} \right\} \notag \\
        &\leq A (1 - \cos{w}) \left( \abs{u}^{\delta} 
        w^{-\alpha-1-\delta} + \abs{u}^{-\delta} w^{-\alpha-1+\delta} \right).
        \label{eq:max_terms}
    \end{align}
    In \eqref{eq:max_terms} the first term is integrable on $[0,\infty)$ since 
    $\alpha+1 + \delta < 3$, and so is the second term since $\alpha+1-\delta>2$.
    Applying Lebesgue's 
    Dominated Convergence Theorem gives \eqref{eq:int_w_SV}. If $w_{0} M_0 > 1$, 
    then we apply similar arguments on the set $\left\{ \abs{u} \beta_{t} 
    \geq w M_0 \right\}$. In either case, there exists $M>0$ such that $1/M \leq w_0$ and 
    \eqref{eq:int_w_SV} holds. Equation \eqref{eq:ch_fn_p2} converging to $0$ 
    as $t\rightarrow 0$ implies (as stated before) that \eqref{eq:ch_fn_p1} 
    converges to $-c_{\alpha} \abs{u}^{\alpha}$. Since $\cos{w} - 1 \leq 0$
    for all $w$, the integrand is non-positive; replacing $\xi^{*}_{S}$ by its
    upper bound from \eqref{eq:xi_equiv} therefore yields a \emph{lower} bound
    on the integral:
    \begin{align}
        \frac{t}{\abs{u} \beta_{t}} \int_{0}^{\infty} &\left( \cos{w} - 1 
        \right) \xi^{*}_{S}\left( \frac{w}{\abs{u} \beta_{t}} \right) 
        \Ind{\frac{w}{\abs{u} \beta_{t}} \leq L} dw  \notag \\
        &\geq \frac{t (1+\varepsilon) \alpha}{\abs{u} \beta_{t}} \int_{0}^{\infty} 
        \left( \cos{w} - 1 \right) \left( \frac{w}{\abs{u} \beta_{t}} 
        \right)^{-\alpha-1} \ell\left( \frac{\abs{u} \beta_{t}}{w} \right) 
        \Ind{\frac{w}{\abs{u} \beta_{t}} \leq L} dw  \notag \\
        &=  (1+\varepsilon) \alpha \abs{u}^{\alpha} t \beta_{t}^{\alpha} 
        \int_{0}^{\infty} \frac{\left( \cos{w} - 1 \right)}{w^{\alpha+1}} 
        \ell\left( \frac{\abs{u} \beta_{t}}{w} \right) 
        \Ind{\frac{w}{\abs{u} \beta_{t}} \leq L} dw  \notag \\
        &=  (1+\varepsilon) \alpha \abs{u}^{\alpha} t \beta_{t}^{\alpha} 
        \ell(\beta_{t}) \int_{0}^{\infty}  \frac{\left( \cos{w} - 1 
        \right)}{w^{\alpha+1}} \frac{\ell\left( \frac{\abs{u} \beta_{t}}{w} 
        \right)}{\ell(\beta_{t})} \Ind{\frac{w}{\abs{u} \beta_{t}} \leq 
        L} dw. \label{eq:lower_bound}
    \end{align}
    Similarly, replacing $\xi^{*}_{S}$ by its lower bound gives the upper bound
    \begin{align}
        \frac{t}{\abs{u} \beta_{t}} \int_{0}^{\infty} &\left( \cos{w} - 1 
        \right) \xi^{*}_{S}\left( \frac{w}{\abs{u} \beta_{t}} \right) 
        \Ind{\frac{w}{\abs{u} \beta_{t}} \leq L} dw  \notag \\
        &\leq  (1-\varepsilon) \alpha \abs{u}^{\alpha} t \beta_{t}^{\alpha} \ell(
        \beta_{t}) 
        \int_{0}^{\infty} \frac{\left( \cos{w} - 1 \right)}{w^{\alpha+1}} 
        \frac{\ell\left( \frac{\abs{u} \beta_{t}}{w} \right)}{\ell(\beta_{t})} 
        \Ind{\frac{w}{\abs{u} \beta_{t}} \leq L} dw.
        \label{eq:upper_bound}
    \end{align}
    Dividing each side of both \eqref{eq:lower_bound} and \eqref{eq:upper_bound} 
    by $t\beta_{t}^{\alpha} \ell(\beta_{t})$, letting $-\varsigma 
    = \int_{0}^{\infty} 
    (\cos{(w)}-1)w^{-\alpha-1}dw$, and letting $t\rightarrow 0$ gives 
    \[
        -(1+\varepsilon) \alpha \varsigma \abs{u}^{\alpha} \leq  \frac{-c_{\alpha} 
        \abs{u}^{ \alpha}}{ \lim_{t\rightarrow 0}t \beta_{t}^{\alpha} \ell(\beta_{t})} 
        \leq -(1- \varepsilon) \alpha \varsigma \abs{u}^{\alpha}.
    \]
    Letting $\varepsilon \rightarrow 0$ implies that  
    \begin{align}
        \lim_{t \rightarrow 0} t \beta_{t}^{\alpha} \ell(\beta_{t}) = \Lambda \in 
        (0,\infty),    \label{eq:coefficients}
    \end{align}
    where $\Lambda = c_{\alpha} / (\alpha \varsigma)$ (note that $\varsigma >0$, for $1<\alpha<2$).
    In fact, e.g. see \cite{sato}, 
    \[
        \varsigma = \frac{\pi}{2 \Gamma\left( 1 + \alpha \right) \sin{\left( 
            \frac{\pi \alpha}{2} \right)}},
    \]
    where $\Gamma$ here is Euler's Gamma function.
\end{proof}

\begin{lem}
 \label{lem:prev_assump_5} 
Let the share--measure L\'evy measure $\nu^{*}$ satisfy
\[
\int_{|x|>1} |x|\,\nu^*(dx) = \int_{|x|>1} |x| e^{x}\,\nu(dx) < \infty.
\]
Then, for every $R_{0}>0$,
\[
\int_{|y|>R_{0}} \gamma^{*}(y)\,dy < \infty,
\]
where $\gamma^{*}(y):=\nu^{*}\big(\{|x|>y\}\big)$.
\end{lem}

\begin{proof}
By symmetry of the absolute value, it suffices to show 
$2\int_{R_{0}}^{\infty}\gamma^{*}(y)\,dy<\infty$. 
Tonelli’s theorem gives, for any $R_{0}>0$,
\[
    \int_{R_{0}}^{\infty}\gamma^{*}(y)\,dy
    = \int_{R_{0}}^{\infty} \int_{\{|z|>y\}} \nu^{*}(dz)\,dy
    = \int_{\{|z|>R_{0}\}} (|z|-R_{0})\,\nu^{*}(dz).
\]
Since $\nu^{*}(dz)=e^{z}\nu(dz)$,
\[
    \int_{\{|z|>R_{0}\}} (|z|-R_{0})\,\nu^{*}(dz)
    \le \int_{\{|z|>R_{0}\}} |z|\,e^{z}\,\nu(dz)
    = I_{1} + I_{2},
\]
with
\[
    I_{1}:=\int_{\{|z|>1\}} |z|\,e^{z}\,\nu(dz), 
    \qquad 
    I_{2}:=\int_{R_{0}<|z|\le 1} |z|\,e^{z}\,\nu(dz).
\]
By \eqref{eq:ex_moment}, $I_{1}<\infty$. 
For $I_{2}$, note that 
\[
    I_{2}\le e \int_{R_{0}<|z|\le 1} \nu(dz)<\infty.
\]
Therefore $\int_{R_{0}}^{\infty}\gamma^{*}(y)\,dy<\infty$, and the claim 
follows.
\end{proof}

One important quantity in the estimation of call option prices in exponential 
\levy models is the expression $\Prob{X_{t} \geq y}$ where $y \geq 0$ and 
$\left( X_{t} \right)_{t \geq 0}$ is a \levy process on $\R$ 
with triplet $\ltrip$. For estimating these tail quantities,
we need concentration inequalities similar to those found in \cite{bhp} and \cite{houdre}. 
The next lemma provides the estimation we need. For now, we will allow $\gamma$ to be 
any nonnegative function, but our future application of this lemma will use the definition 
of $\gamma$ in \eqref{eq:levy_tail}; however, for the proof of the main theorem 
we do require the function $V$ defined in \eqref{eq:var_def} and $\mu_{y}$ defined in 
\eqref{eq:mu_epsilon}.

\begin{lem}
    \label{lem:concentration}
    Let $\gamma: \R^{+} \rightarrow \R^{+}$ be such that for all $R>0$,
    \begin{enumerate}[label=(\roman*)]
        \item $\int_{\abs{x}>R} \nu\left( dx \right) \leq \gamma\left( R 
            \right)$, 
        \item and there exists $C>0$ independent of $R$ such that $V(R) \leq 
            CR^{2} \gamma\left( R \right)$.
    \end{enumerate}
    Then, for every $y>0$ and for every $0 < t < y / 4 \left( \mu_{y/4} 
    \right)_{+}$ (with $y/0 = \infty$), 
    \[
        \Prob{ X_t \geq y } \leq \left( 1 + Ce^2 \right) t \gamma\left( 
        \frac{y}{4} \right).
    \]
\end{lem}

\begin{proof}
    In the traditional manner (e.g. see \cite{bhp} or \cite{houdre_marchal}), we 
    break $X = \proc{X_t}$ up into two parts, 
    $X^{\varepsilon} = \proc{X_t^{ \varepsilon}}$ which consists of all jumps 
    smaller than $\varepsilon$ and $\widetilde{X}^{\varepsilon} = 
    \proc{\widetilde{X}_t^{\varepsilon}}$ consisting of all jumps larger than 
$\varepsilon$. For each $t>0$, we can represent $X_{t}$ as
\begin{align}
    X_{t} = b t  + \int_{0}^{t} \int_{\abs{x} \leq 1} x \, 
    \tilde{N}(dx, ds) + \int_{0}^{t} \int_{\abs{x} \geq 1} 
    x \, N(dx, ds)
    \label{eq:levy_ito}
\end{align}
where $N$ is a Poisson random measure on $\R \backslash \left\{ 0 \right\}$
with intensity $\nu(dx)\, dt$ and $\tilde{N}$ is its compensated version. Let $f_{\varepsilon}(x) = 
\1_{\left[ -\varepsilon, \varepsilon \right] }$ and $\bar{f}_{\varepsilon} =
1 - f_{\varepsilon}$. We can define the processes for each $t>0$ by
\begin{align}
    \widetilde{X}_{t}^{\varepsilon} = \int_{0}^{t} \int_{\R} x \bar{f}_{
    \varepsilon}(x) N(dx, ds) \mbox{\;\; and \;\; } X^{\varepsilon}_{t} =
    X_{t} - \widetilde{X}_{t}^{\varepsilon}.
    \label{eq:small_and_big_jumps_def}
\end{align}
The process $\widetilde{X}^{\varepsilon}$ is a compound Poisson process with 
intensity $\lambda_{\varepsilon} = \int \bar{f}_{\varepsilon} (x) \nu(dx)$ 
and jump distribution 
\[
    \frac{\bar{f}_{\varepsilon}(x) \nu (dx)}{\lambda_{\varepsilon}},
\]
and $X^{\varepsilon}$ is a \levy process with characteristic triplet 
$\left( b_{\varepsilon}, 0,  f_{\varepsilon} \nu \right)$ where
\[
    b_{\varepsilon} = b - \int_{\abs{x}\leq 1} x \bar{f}_{\varepsilon}(x) 
    \nu(dx).
\]
We will need the fact that $\Expecnb{
X_t^{\varepsilon}} = t \mu_{\varepsilon}$ where $\mu_{\varepsilon}$
is defined by \eqref{eq:mu_epsilon}. For a fixed $y>0$, we have
\begin{align}
    \Prob{X_t \geq y} &\leq \Prob{X_t^{\varepsilon} \geq y/2} + 
    \Prob{\widetilde{X}_t^{\varepsilon} \geq y/2} \notag \\
    &\leq \Prob{X_t^{\varepsilon} - \Expecnb{X_t^{\varepsilon}} \geq y/2 - 
    \Expecnb{X_t^{\varepsilon}}} + 
    \Prob{\widetilde{X}_t^{\varepsilon} \ne 0} \notag \\
    &\leq  \Prob{X_t^{\varepsilon} - \Expecnb{X_t^{\varepsilon}} \geq y/2 - 
    \Expecnb{X_t^{\varepsilon}}} + t \gamma(\varepsilon) \notag \\
    &=  \Prob{X_t^{\varepsilon} - \Expecnb{X_t^{\varepsilon}} \geq y/2 - t 
    \mu_{\varepsilon}} + t \gamma(\varepsilon). \label{eq:2parts}
\end{align}
Using a general concentration inequality (e.g. Corollary 1 in \cite{houdre})
and writing $V_{\varepsilon}:=V(\varepsilon)$, 
we obtain for $z>0$
\begin{align}
    \Prob{X_t^{\varepsilon} - \Expecnb{X_t^{\varepsilon}} \geq z} &\leq 
    \exp{\left[ \frac{z}{\varepsilon}
    - \left( \frac{z}{\varepsilon} + \frac{tV_{\varepsilon}}{\varepsilon^{2}} 
    \right) 
    \log{\left( 1 +
    \frac{\varepsilon z}{t V_{\varepsilon}} \right)} \right] } \notag \\
    &\leq \exp{\left[ \frac{z}{\varepsilon} - \frac{z}{\varepsilon} 
        \log{\left( 1 + \frac{\varepsilon z}{C t \varepsilon^{2} 
    \gamma\left( \varepsilon \right) } \right)} \right] } \notag \\ 
    &= \exp{\left[ \frac{z}{\varepsilon} - \frac{z}{\varepsilon} \log{\left( 
        1 + \frac{z}{C t \varepsilon 
    \gamma\left( \varepsilon \right) } \right)} \right] }
    = \frac{\exp{\left( \frac{z}{\varepsilon} \right)}}{\left( 
        1 + \frac{z}{C t \varepsilon \gamma(
    \varepsilon) } \right)^{z/ \varepsilon} } .
    \label{eq:zineq}
\end{align}
We now need to choose $\varepsilon$ in such a way that both terms 
in \eqref{eq:2parts} are of the same order. We choose 
$\varepsilon = y/4$ and we consider two cases. First, consider the case 
where $\mu_{y/4} \geq 0$ and further assume that
$0 < t < y / 4 \mu_{y/4}$ (equivalently $y/2 - t \mu_{y/4} > y/4$). Then, 
\begin{align}
    \Prob{X_t^{y/4} - \Expecnb{X_t^{y/4}} \geq y/2 - t \mu_{y/4}} &\leq
    \Prob{X_t^{y/4} - \Expecnb{X_t^{y/4}} \geq y/4 } \notag \\
    &= \frac{\exp{\left( \frac{y/4}{y/4} \right)}}{\left( 1 + \frac{y/4}{C t 
        (y/4) \gamma( y/4) } \right)^{\frac{y/4}{y/4} } }  \notag \\
    &= \frac{e}{\left( 1 + \frac{1}{C t \gamma(y/4) } \right) } 
    \leq Ce^2 t \gamma\left( \frac{y}{4} \right).
    \label{eq:pos_case}
\end{align}    
Next, consider the case where $\mu_{y/4} <0$. Then, for all $t>0$,
\begin{align}
    \Prob{X_t^{y/4} - \Expecnb{X_t^{y/4}} \geq y/2 - t \mu_{y/4}} &\leq 
    \Prob{X_t^{y/4} - \Expecnb{X_t^{y/4}} \geq y/2 }  \notag \\
    &\leq \frac{\exp{\left( \frac{y/2}{y/4} \right)}}{\left( 1 + \frac{
    y/2}{C t (y/4) \gamma(y/4) } \right)^{\frac{y/2}{y/4}} } \notag \\
    &= \frac{e^2}{\left( 1 + \frac{2}{C t \gamma(y/4) } \right)^{2}} 
    \leq \frac{e^2}{\left( 1 + \frac{2}{C t \gamma(y/4) } \right)}
    \leq Ce^2 t \gamma\left( \frac{y}{4} \right).
    \label{eq:neg_case}
\end{align}
Notice that the terms in \eqref{eq:pos_case} and \eqref{eq:neg_case} are the 
same. Combining this result with \eqref{eq:2parts} gives
\[
    \Prob{X_t \geq y} \leq Ce^2 t \gamma\left( \frac{y}{4} \right) + t 
    \gamma\left( \frac{y}{4} \right) = (1+Ce^2)t 
    \gamma\left( \frac{y}{4} \right),
\]
for all $y>0$ and $0< t < y / 4 (\mu_{y/4})_{+}$. 
\end{proof}

\begin{prop}
\label{prop:VRatio}
    Let $\left( X_{t} \right)_{t \geq 0}$ be a \levy process in the domain of attraction of an 
    $\alpha$-stable random variable with $\alpha \in (1,2)$. 
    Then for any $0<x<y< \infty$, 
    \[
        \sup_{x\leq R \leq y} \frac{V(R)}{R^{2} \gamma(R)} < \infty.
    \]
\end{prop}

\begin{proof}
    By Theorem~\ref{thm:maller_mason}, $\gamma(R) = R^{-\alpha} 
    \ell(R)$ where $\ell$ is slowly varying and $\alpha \in (1,2)$.
    Integration by parts gives
    \begin{align}
    0 < V(z) &= -z^{2} \gamma(z) + 2 \int_{0}^{z} y \gamma(y)\, dy
        = -z^{2} \gamma(z) + 2 \int_{0}^{z} y^{1-\alpha} \ell\left( y \right)
        \,dy \label{eq:V_ibp} 
    \end{align}
    which is well-defined since $1-\alpha \in (-1, 0)$. So, 
    \begin{align}
        \frac{V(z)}{z^{2} \gamma\left( z \right)} &= 
        \frac{ -z^{2} \gamma(z) + 2 \int_{0}^{z} y^{1-\alpha} \ell\left( y 
        \right)\, dy}{z^{2} \gamma(z)}
        = -1 + \frac{2 \int_{0}^{z} y^{1-\alpha} \ell\left( y \right)\, dy}{
            z^{2-\alpha} \ell\left( z \right)}. \notag
    \end{align}
    The numerator is continuous and the denominator is piecewise continuous and
    bounded away from $0$ in any compact interval of $\R_{+}$ not including
    $0$ (the function $\gamma$ is nonincreasing on $(0, \infty)$ so it can only
    have jump discontinuities). Thus, the supremum is bounded and the result 
    follows. 
\end{proof}

\begin{prop}
    \label{prop:concentration_ineq}
    Let $\left( X_{t} \right)_{t \geq 0}$ be a \levy process in the domain of attraction of an 
    $\alpha$-stable random variable, $0< \alpha < 2$. 
    Further, let there exist $R_{0} > 0$ and $C>0$ possibly depending on $R_{0}$ such that for 
    all $R> R_{0}$, $V(R) \leq C R^{2} \gamma\left( R \right)$ where 
    $V$ and $\gamma$ are defined in \eqref{eq:var_def} and 
    \eqref{eq:levy_tail}, respectively. Then, for every $y>0$ and for every 
    $0 < t < y / 4 \left( \mu_{y/4} \right)_{+}$ (with $y/0 = \infty$), 
    \[
        \Prob{X_{t} \geq y} \leq \left( 1 + C e^{2} \right) t \gamma \left( 
        \frac{y}{4} \right).
    \]
\end{prop}

\begin{proof}
By assumption, $X$ is in the domain of attraction of an $\alpha$-stable 
random variable. When $\gamma$ is defined as in \eqref{eq:levy_tail}, 
condition \emph{(i)} of Lemma~\ref{lem:concentration} is then satisfied 
trivially.

Moreover, it is shown in \cite{bingham} (Chapter~8.1), \cite{feller} (VII.9), 
and \cite{maller_mason} that for processes in the domain of attraction of an 
$\alpha$-stable law, condition \emph{(ii)} of Lemma~\ref{lem:concentration}
holds true automatically for small $R>0$: there exist $R_{1}>0$ and $C_{1}>0$ such
that
\[
    V(R) \leq C_{1} R^{2} \gamma(R), \qquad 0 < R \leq R_{1}.
\]

On the other hand, Proposition~\ref{prop:VRatio} above shows that for any 
$0<x<y<\infty$,
\[
    \sup_{x \le R \le y} \frac{V(R)}{R^{2}\gamma(R)} < \infty.
\]
In particular, choosing $x:=R_{1}$ and $y:=R_{0}$ we obtain a constant 
$C_{2}>0$ such that
\[
    V(R) \le C_{2} R^{2}\gamma(R), \qquad R_{1} \le R \le R_{0}.
\]

By hypothesis of the present proposition, there exists $C>0$ such that
\[
    V(R) \le C R^{2} \gamma(R), \qquad R>R_{0}.
\]
Combining these three bounds and letting
\[
    \widetilde{C} := \max\{C_{1}, C_{2}, C\},
\]
we get a uniform estimate
\[
    V(R) \le \widetilde{C}\,R^{2}\gamma(R), \qquad \text{for all } R>0.
\]
Thus both conditions \emph{(i)} and \emph{(ii)} of Lemma~\ref{lem:concentration} 
hold with the same constant $\widetilde{C}$. Applying that lemma with 
this constant $\widetilde{C}$ yields, for every $y>0$ and every 
$0<t<y/4\,(\mu_{y/4})_{+}$,
\[
    \Prob{X_{t} \ge y} \le \left(1 + \widetilde{C} e^{2}\right) 
    t\,\gamma\!\left(\frac{y}{4}\right).
\]
\end{proof}

\begin{proof}[Proof of Theorem \ref{thm:main_result}]
    Recalling \eqref{eq:cm_rep} and $\beta_{t} = 1/B_t$ and using $M_0$ from 
    the Potter bound argument in Theorem \ref{thm:beta_behavior}, we obtain
    \begin{align}
        \frac{ c(t, 0) }{ B_t } &= \frac{1}{B_{t}} \int_{0}^{\infty}e^{-x}
        \Probs{X_{t} \geq x} dx \notag \\
        &= \int_{0}^{\infty} e^{-B_{t} u } \Probs{X_{t} \geq B_{t} u} du 
        \notag \\
        &= \int_{0}^{\infty} e^{-B_{t} u } \Probs{X_{t} \geq B_{t} u} 
        \left( \Ind{\frac{1}{M_{0}} \geq \frac{B_t u}{4} \geq t \bar{\mu}^{*}} 
        + \Ind{\frac{B_{t} u}{4} > \frac{1}{M_{0}}} + \Ind{\frac{B_{t} u}{4} 
            < t \bar{\mu}^{*}} 
        \right) du \label{eq:integral_rep} \\
        &:=\int_{0}^{\infty} \left( A_1(t,u) + A_2(t,u) + A_3(t,u) 
        \right) du, \label{eq:three_Aterms}
    \end{align}
where $t$ is so small that $t \bar{\mu}^{*} < 1 / M_0$.
In what follows, we will write $A_{i}(t) := A_{i}(t,u)$ for $i = 1, 2, 3$. 

First, we note that the integral of $A_3(t)$ can be estimated as
\begin{align}
    \int_{0}^{\infty} e^{-B_{t} u } \Probs{X_{t} \geq B_{t} u} 
    \Ind{\frac{B_{t} u}{4} < t \bar{\mu}^{*}} du  &\leq
    \int_{0}^{\infty} \Ind{u < 4 \bar{\mu}^{*} t \beta_{t}} du
    = 4 \bar{\mu}^{*} t \beta_{t}, \label{eq:A3_term}
\end{align}
as $t\rightarrow 0$.  Indeed, from Theorem~\ref{thm:beta_behavior} we know that
$t \beta_{t}^{\alpha} \ell\left( \beta_{t} \right) \sim \Lambda_{\alpha}$ and
$\beta_t\to\infty$, so
\[
    t\beta_t
    = \frac{t\beta_t^{\alpha}\ell(\beta_t)}{\beta_t^{\alpha-1}\ell(\beta_t)}
    \sim \frac{\Lambda_{\alpha}}{\beta_t^{\alpha-1}\ell(\beta_t)} \longrightarrow 0,
\]
because $\alpha>1$ and $\ell$ is slowly varying.  Therefore we only need to deal with 
the integral of $A_1(t)$ and $A_2(t)$.

Using $\{ B_{t} u > 4 t \bar{\mu}^{*} \} 
\subseteq \{ B_{t} u > 4 t \mu^{*}_{B_{t} u/4} \}$
and the estimate from Proposition \ref{prop:concentration_ineq}, 
for some constant $C>0$ and for any 
\begin{align}
    u \in \mathcal{I} \subset \{ B_{t} u > 4 t \bar{\mu}^{*} \} \label{eq:I_set}
\end{align}
with $\mathcal{I}$ measurable and $t>0$ fixed, 
\begin{align}
    \Probs{X_t \geq B_{t} u} &\leq \left[ \left( 1 + C e^{2} \right) t 
        \gamma^{*}\left( \frac{B_{t} u}{4} \right) \wedge 1 \right]
        \notag \\
    &= \left[ \left( 1 + C e^{2} \right) t \gamma^{*}\left( \frac{u}{4\beta_t} 
        \right) \wedge 1 \right] \notag \\
    &\leq \left[ \left( (1 + C e^{2}) t \left( \frac{u}{4 \beta_t}
        \right)^{-\alpha} \ell\left( \frac{4 \beta_t}{u} \right) \right) \wedge 
        1 \right]
    = \left[ \kappa t \beta_{t}^{\alpha} u^{-\alpha} \ell\left( 
        \frac{4 \beta_t}{u} \right) \wedge 1 \right],
        \label{eq:two_subsets_est} 
\end{align}
where $\kappa>0$ is a collection of all the constants. In what follows, we use
$\kappa$ to represent a positive constant whose value might change from line
to line.

Recall from our argument in Theorem \ref{thm:beta_behavior} that 
$\alpha \pm \delta \in (1,2)$. We also choose $t_0>0$ such that for all $0<t<t_0$,
\[
    \frac{\Lambda}{2\ell(\beta_t)} < t \beta_{t}^{\alpha} < 
    \frac{3\Lambda}{2\ell(\beta_{t})}.
\]
Continuing \eqref{eq:two_subsets_est} for $0<t< t_0$ and $u \in \mathcal{I}$, 
\begin{align}
    \Probs{X_t \geq B_{t} u} &\leq \kappa u^{-\alpha} 
        \frac{\ell\left( \frac{4\beta_t}{u} \right)}{\ell(\beta_t)} \wedge 1. \notag 
\end{align}

First, we show that the integral of $A_2(t)$ goes to $0$ as $t\rightarrow 0$. 
Indeed, choosing $\mathcal{I} = \left\{ 4\beta_{t} < M_0 u \right\}$ in 
\eqref{eq:I_set} and changing variables gives
\begin{align}
    \int_{0}^{\infty} A_2(t) du &\leq
    \int_{0}^{\infty} \left[ \kappa u^{-\alpha} \frac{\ell\left( \frac{4
    \beta_t}{u} \right)}{\ell(\beta_t)} \wedge 1 \right]
    \Ind{\frac{4\beta_t}{u} < M_0 } du \notag \\
    &\leq \frac{1}{4 \beta_{t}} \int_{0}^{\infty} \kappa u^{2-\alpha} 
    \frac{\ell\left( \frac{4\beta_t}{u} \right)}{\ell(\beta_t)} 
    \Ind{\frac{4\beta_t}{u} < M_0 } 4 \beta_{t} \frac{du}{u^{2}} \notag \\
    &= \frac{1}{4 \beta_{t}} \int_{0}^{M_0} \kappa \left( \frac{4 \beta_t}{w} 
    \right)^{2-\alpha} \frac{\ell(w)}{\ell(\beta_t)} 
        dw \notag \\
    &= \frac{\kappa}{\beta_{t}^{\alpha-1} \ell(\beta_{t})} \int_{0}^{M_0} 
    \frac{\ell(w)}{w^{2-\alpha}} dw \notag \\
    &=  \frac{\kappa}{\beta_{t}^{\alpha-1} \ell(\beta_{t})} 
    \int_{1/M_{0}}^{\infty} \frac{\ell\left( \frac{1}{z} \right)}{z^{\alpha}} 
    dz
    = \frac{\kappa}{\beta_{t}^{\alpha-1} \ell(\beta_{t})} 
    \int_{1/M_{0}}^{\infty} \gamma^{*}(z) dz \rightarrow 0,
    \label{eq:goes_to_zero}
\end{align}
as $t \rightarrow 0$, since the integral is finite and $t \beta_{t}^{\alpha} 
\ell(\beta_{t}) \sim \Lambda$ as $t \rightarrow 0$. In order to estimate 
$A_{1}\left( t \right)$, we use \eqref{eq:two_subsets_est} and Potter bounds 
from Theorem \ref{thm:beta_behavior}. Choosing $\mathcal{I} = \left\{ M_0 
    \leq \frac{4 \beta_{t}}{u} 
< \frac{1}{t \bar{\mu}^{*} } \right\}$ and for $u \in \mathcal{I}$ and any $t > 0$, 
\begin{align}
    A_{1}(t, u) &\leq \left[ \kappa u^{-\alpha} \frac{\ell\left( \frac{4\beta_t}{u} 
    \right)}{\ell(\beta_t)} \wedge 1 \right]
    \Ind{M_0 \leq \frac{4\beta_t}{u} < \frac{1}{t \bar{\mu}^{*}} } \notag \\
    &\leq \max{ \left\{ \kappa 4^{\delta} u^{-\alpha-\delta} \wedge 1, 
    \kappa 4^{-\delta} u^{-\alpha+\delta} \wedge 1 \right\} } \in 
    L^{1}([0,\infty)), \label{eq:A1_term}
\end{align}
so that we are able to apply Lebesgue's Dominated Convergence 
theorem to 
\[
    \int_{0}^{\infty} A_{1}\left( t,u \right) du = \int_{0}^{\infty} e^{-B_t u} 
    \Probs{X_t \geq B_t u} \Ind{\frac{1}{M_0} \geq 
    \frac{B_t u}{4} > t \bar{\mu}^{*}} du.
\]

Combining this fact with \eqref{eq:goes_to_zero} and \eqref{eq:A1_term} gives
\begin{align*}
    \lim_{t\rightarrow 0} \frac{c\left( t, 0 \right)}{ B_t } 
    &= \lim_{t \rightarrow 0} \int_{0}^{\infty} A_{1}\left( t, u \right) du \\
    &= \int_{0}^{\infty} \Probs{Z\geq u} du \\
    &= \Expecsnb{Z_{+}},
\end{align*}
which can be rewritten as 
\[
   \Expecnb{\left( S_{t} - S_{0} \right)_{+} } = 
    S_{0} B_{t} \Expecsnb{Z_{+}} + o\left( B_{t} \right),
\]
as $t\to 0$, proving the theorem.
\end{proof}

\begin{proof}[Proof of Corollary \ref{cor:no_brownian_case}]
    We proceed as in the proof of Proposition 3.7 in 
    \cite{fl_houdre_cgmy}. 
    We know that the Black-Scholes call price asymptotics are $S_{0} c_{BS}$, 
    where
    \begin{align}
        c_{BS}\left(t, \sigma\right) = \frac{\sigma}{\sqrt{2 \pi}} \sqrt{t} 
        + o\left( \sqrt{t} \right), \label{eq:BS_asymptotics}
    \end{align}
    as $t \rightarrow 0$, since $c_{BS} = N\left( \sigma \sqrt{t} \right)$ where
    \[
    N\left( \theta \right) := 
        \int_{0}^{\theta} \Phi^{'}\left( \frac{u}{2} \right) du =
        \frac{1}{\sqrt{2 \pi}} \int_{0}^{\theta} \exp{\left( -\frac{u^{2}}{8}
        \right)} du,
    \]
    where $\Phi$ is the standard normal cumulative distribution function,
    and where $N$ has asymptotic behavior
    \[
        N\left( \theta \right) = \frac{1}{\sqrt{2 \pi}} \theta + o\left( \theta
            \right),
    \]
    as $\theta \rightarrow 0$. We need an expression similar to 
    \eqref{eq:BS_asymptotics} where the constant $\sigma$ is replaced by 
    the implied volatility function $\hat{\sigma}(t)$. Now, $\hat{\sigma}(t) 
    \rightarrow 0$ as $t \rightarrow 0$, so a substitution in $c_{BS}$ gives
    \begin{align}
        c_{BS}\left( t, \hat{\sigma}(t) \right) &= \frac{\hat{\sigma}(t)}{
            \sqrt{2 \pi}} \sqrt{t} + o\left( \hat{\sigma}\left( t \right)
            \sqrt{t} \right) 
            = \frac{\hat{\sigma}(t)}{
            \sqrt{2 \pi}} \sqrt{t} + o\left( \sqrt{t} \right), \label{eq:BS_sighat}
    \end{align}
    as $t \rightarrow 0$, since $\hat{\sigma}\left( t \right) = o\left( 1 \right)$ 
    as $t \rightarrow 0$. Equating \eqref{eq:BS_sighat} with 
    \eqref{eq:call_price_nobs}, leads to
    \[
        \frac{\hat{\sigma}(t)}{\sqrt{2 \pi}} \sqrt{t} \sim B_{t} \Expecsnb{Z_{+}},
    \]
    as $t\rightarrow 0$, i.e.
    \[
        \hat{\sigma}\left( t \right) \sim \sqrt{2 \pi} \frac{B_{t}}{\sqrt{t}} 
        \Expecsnb{Z_{+}},
    \]
    as $t \rightarrow 0$, giving the result. 
\end{proof}

\begin{lem}
    \label{EDCT}
    Let $(S,\Sigma,\mu)$ be a measure space. Let $f,g:S\times [ 0,\infty) 
        \rightarrow [ 0,\infty)$ and $h:S\rightarrow [0,\infty)$ be measurable
            and such that $f\left( \cdot, t \right), g\left( \cdot, t \right), 
            h \in L^{1}\left( S \right)$ for almost every $t \geq 0$. Also, suppose
            \begin{enumerate}[label=(C\arabic*)]
                \item \label{C2} $f(s,t) \rightarrow \bar{f}(s) \in L^{1}(S)$ 
                    as $t \rightarrow 0$,
                \item \label{C3} $f(s,t) \leq h(s) + g(s,t)$ for $\mu$-a.e. $s\in S$ 
                    and almost every $t\geq 0$,
                \item \label{C5} $g(s,t) \rightarrow 0$ as $t\rightarrow 0$ 
                    $\mu$-a.e $s \in S$, 
                \item \label{C6} $\int_{S} g(s,t) \mu(ds) \rightarrow 0$ as 
                    $t\rightarrow 0$.
            \end{enumerate}
            Then 
            \[
                \lim_{t\rightarrow 0} \int_{S} f(s,t) \mu(ds) = \int_{S} 
                \bar{f}(s) \mu(ds).
            \]
\end{lem}

\begin{proof}
    First, we choose a sequence $\procn{t_{n}}$ such that $t_{n} \rightarrow 0$,
    as $n \rightarrow \infty$, and
    \begin{align}
        \lim_{n \rightarrow \infty} \int_{S} f\left( s, t_{n} \right) \mu\left( 
        ds \right) = \liminf_{t \rightarrow 0} \int_{S} f\left( s, t \right) 
        \mu\left( ds \right).
    \end{align}
    We apply Fatou's lemma to the nonnegative functions $f(\cdot,t_n)$, since 
    $f \geq 0$, to obtain
    \begin{align}
        \int_{S} \bar{f}(s) \mu(ds) &=
        \int_{S} \liminf_{t \rightarrow 0} f(s,t) \mu(ds) \notag \\
        &\leq \int_{S} \liminf_{n \rightarrow \infty} f(s,t_n) \mu(ds)
        \leq \liminf_{n \rightarrow \infty} \int_{S} f(s,t_n) \mu(ds)
        = \liminf_{t\rightarrow 0} \int_{S} f(s,t) \mu(ds). 
        \label{eq:liminf}
    \end{align}
    Next, we choose another sequence $\procn{t_{n}^{'}}$ such that 
    $t_{n}^{'} \rightarrow 0$, as $n \rightarrow \infty$, and
    \begin{align}
        \lim_{n \rightarrow \infty} \int_{S} f\left( s, t_{n}^{'} \right) \mu\left( 
        ds \right) = \limsup_{t \rightarrow 0} \int_{S} f\left( s, t \right) 
        \mu\left( ds \right).
    \end{align}
    So, \ref{C3} implies $h(s) + g(s,t) - f(s,t) \geq 0$, and 
    applying Fatou's lemma again,
    \begin{align*}
        \int_{S} \left( h(s) - \bar{f}(s) \right) \mu(ds) &= 
        \int_{S} \liminf_{t \rightarrow 0} \left( h(s) + g\left( s,t \right) 
        - f\left( s, t \right) \right) \mu\left( ds \right) \\
        &\leq \int_{S} \liminf_{n \rightarrow \infty} \left( h(s) + g\left( s,
        t_{n}^{'} \right) 
        - f\left( s, t_{n}^{'} \right) \right) \mu\left( ds \right) \\
        &\leq \liminf_{n \rightarrow \infty} \int_{S} \left( h(s) + g\left( s,
        t_{n}^{'} \right) 
        - f\left( s, t_{n}^{'} \right) \right) \mu\left( ds \right) \\
        &= \int_{S} h(s) \mu\left( ds \right) + \liminf_{n \rightarrow \infty}
        \left( - \int_{S} f\left( s, t_{n}^{'} \right) \mu\left( ds \right) \right) \\
        &= \int_{S} h(s) \mu(ds) - \limsup_{n\rightarrow \infty} \int_{S} f(s,t_{n}^{'}) 
        \mu(ds) \\
        &= \int_{S} h(s) \mu(ds) - \limsup_{t\rightarrow 0} \int_{S} f(s,t) 
        \mu(ds).
    \end{align*}
    Note that the limsup in the last line is finite due to \ref{C3} and the
    integrability of $h$ and $g(\cdot, t)$. Canceling 
    the $h$ term, we obtain 
    \begin{align}
        \limsup_{t\rightarrow 0} \int_{S} f(s,t) \mu(ds) \leq \int_{S} 
        \bar{f}(s) \mu(ds). \label{eq:limsup}
    \end{align}
    Combining \eqref{eq:liminf} and \eqref{eq:limsup}, we have
    \[
        \int_{S} \bar{f}(s) \mu(ds) \leq \liminf_{t\rightarrow 0} \int_{S} 
        f(s,t) \mu(ds) \leq \limsup_{t\rightarrow 0} \int_{S} f(s,t) \mu(ds) 
        \leq \int_{S} \bar{f}(s) \mu(ds),
    \]
    which proves the result.
\end{proof}

\begin{proof}[Proof of Theorem~\ref{thm:Brownian}]
    We make use of the fact that
    \begin{align}
        \lim_{t\rightarrow 0} \Expecs{\exp{\left( i u \frac{X_t}{\sqrt{t}} 
        \right) } } = \exp{\left(-\frac{1}{2} \sigma^{2} u^{2}\right)},
        \label{eq:ch_fn_convergence}
    \end{align}
    and
    \begin{align}
        \lim_{t\rightarrow 0} \Probs{\frac{X_{t}}{\sqrt{t}} \geq x } = 
        \Probs{\sigma W_{1} \geq x}. 
        \label{eq:prob_convergence}
    \end{align}
    Continuing, 
    \begin{align}
        \frac{1}{\sqrt{t}} \Expecnb{(S_{t}-S_{0})_{+}} &= \frac{1}{\sqrt{t}} 
        \int_{0}^{\infty} e^{-x} \Probs{X_t \geq x} dx
        = \int_{0}^{\infty} e^{-\sqrt{t} u} \Probs{X_t \geq u \sqrt{t}} du. 
        \label{eq:main_int2}
    \end{align}
    Note that 
    \begin{align}
        e^{-\sqrt{t} u} \Probs{X_{t} \geq u \sqrt{t} } \leq \;\; 
        &\Probs{\sigma W_{t} \geq \frac{\sqrt{t} u}{2}} 
        + \Probs{L_{t} \geq \frac{\sqrt{t} u}{2} }
        = \Probs{\sigma W_{1} \geq \frac{u}{2}}  + 
        \Probs{L_{t} \geq \frac{\sqrt{t} u}{2} }. \label{eq:both_parts}
    \end{align}
    We next use Lemma~\ref{EDCT} with 
    \begin{align*}
        f(u,t) &= e^{-\sqrt{t} u} \Probs{X_{t} \geq u \sqrt{t} }, \\
    \bar{f}(u) &= \Probs{\sigma W_{1} \geq u}, \\
    h(u) &= \Probs{\sigma W_{1} \geq u/2},
    \end{align*}
    and $g(u,t) = \Probs{L_{t} \geq \sqrt{t} u/2}$. It is easy to see that
    conditions \ref{C2} -- \ref{C3} of the lemma are satisfied. We need to show 
    that \ref{C5} and \ref{C6} also hold. It is not too difficult to see
    that \ref{C5} holds from \eqref{eq:ch_fn_convergence} and the 
    independence of $W$ and $L$ under $\Ps$. 
    We now show that \ref{C6} is satisfied as well. Once done, we 
    immediately have the result by applying Lemma \ref{EDCT} to 
    \eqref{eq:main_int2}. First note
    \begin{align*}
        \int_{0}^{\infty} \Probs{L_{t} \geq \frac{\sqrt{t} u}{2} } du &=
        \int_{0}^{\infty} \Probs{L_{t} \geq \frac{\sqrt{t} u}{2} } \left( 
        \Ind{\frac{\sqrt{t} u}{8} \geq t \bar{\mu}^{*} } + 
        \Ind{\frac{\sqrt{t} u}{8} < t \bar{\mu}^{*} } \right) du \\
        &:= D_{1}(t) + D_{2}(t).
    \end{align*}
    It is easy to see that $D_{2}(t) \rightarrow 0$ as $t\rightarrow 0$ since
    \[
        D_{2}(t) \leq \int_{0}^{\infty} \Ind{ \frac{\sqrt{t} u}{8}< t 
        \bar{\mu}^{*} } du = 8 \sqrt{t} \bar{\mu}^{*}.
    \]
    We now show $\lim_{t \rightarrow 0} D_{1}(t)=0$ by making use of 
    Lemma \ref{prop:concentration_ineq} and Potter's bounds. To do so, we 
    need to break up $D_{1}(t)$ into several pieces. Choose $A>1$ and 
    $\delta>0$ such that $\alpha \pm \delta \in (1,2)$ and let $M_0>0$ be 
    such that the Potter bounds hold for $\ell$ on $\left[ \left. M_{0},
        \infty \right) \right.$. We have
    \begin{align*}
        D_{1}(t) &= \int_{0}^{\infty} \Probs{L_{t} \geq \frac{\sqrt{t} 
        u}{2} }\left( \Ind{M_0 \leq \frac{8}{\sqrt{t} u} \leq \frac{1}{t 
        \bar{\mu}^{*}} } + \Ind{\frac{8}{\sqrt{t} u} < M_{0} } \right) du \\
        &:= D_{11}(t) + D_{12}(t),
    \end{align*}
    and we will apply the Potter bounds to $D_{11}(t)$ in order to use a 
    Dominated Convergence Theorem argument. We are concerned with the limit as 
    $t\rightarrow 0$, so there is no loss of generality in assuming that $t$ is 
    so small that $8 > M_0 \sqrt{t}$. 
    We proceed by estimating (and using $\kappa$ as a general positive 
    coefficient that can change from line to line)
    \begin{align}
        D_{11}(t) &\leq \int_{0}^{\infty} \left( (1+Ce^{2})t \gamma^{*}\left( 
        \frac{\sqrt{t} u}{8} \right) \wedge 1 \right)
        \Ind{M_0 \leq \frac{8}{\sqrt{t} u} \leq \frac{1}{t \bar{\mu}^{*}} } du 
        \notag \\
        &= \int_{0}^{\infty} \left( \kappa t \left( \frac{\sqrt{t} u}{8} 
        \right)^{-\alpha} \ell \left( \frac{8}{\sqrt{t} u} \right) \wedge 1 
        \right) \Ind{M_0 \leq \frac{8}{\sqrt{t} u} \leq \frac{1}{t 
            \bar{\mu}^{*}} } du \notag \\
        &= \int_{0}^{\infty} \left( \kappa t^{1-\alpha/2} u^{-\alpha} 
        \ell\left( \frac{8}{\sqrt{t} u} \right) \wedge 1 \right) \Ind{M_0 \leq 
            \frac{8}{\sqrt{t} u} \leq \frac{1}{t \bar{\mu}^{*}} } du. \notag
    \end{align}
    We have 
    \begin{align}
        \left( \kappa t^{1-\alpha/2} u^{-\alpha} \ell\left( \frac{8}{\sqrt{t} 
        u} \right) \wedge 1 \right)
        &\Ind{M_0 \leq \frac{8}{\sqrt{t} u} \leq \frac{1}{t \bar{\mu}^{*}} } 
        \notag \\
        &=\left( \kappa t^{1-\alpha/2} u^{-\alpha} \frac{\ell\left( 
            \frac{8}{\sqrt{t} u} \right) }{\ell\left( \frac{8}{\sqrt{t}} 
            \right)} \ell\left( \frac{8}{\sqrt{t}} \right) \wedge 1 \right)
            \Ind{M_0 \leq \frac{8}{\sqrt{t} u} \leq \frac{1}{t \bar{\mu}^{*}} } 
            \notag \\
        &\leq \kappa t^{1-\alpha/2} \ell\left( \frac{8}{\sqrt{t}} \right) A 
        \max\left( u^{-\alpha-\delta}, u^{-\alpha+\delta} \right) \wedge 1. 
        \label{eq:d11_bound}
    \end{align}
    Now, letting $\beta_{t} = 8 / \sqrt{t}$,
    \begin{align}
        t^{\left(1-\frac{\alpha}{2}\right)} \ell\left( \frac{8}{\sqrt{t}} 
        \right) &= 8^{2-\alpha} \left( \frac{8}{\sqrt{t}} \right)^{2 
            \left(\frac{\alpha}{2}-1 \right)} \ell(\beta_{t})
        = 8^{2-\alpha} \beta_{t}^{\alpha-2} \ell(\beta_{t}) \rightarrow 0,
        \notag 
    \end{align}
    since $\alpha-2<0$ and $\beta_{t} \rightarrow \infty$ as $t\rightarrow 0$. 
    Thus, there exists $t_0$ such that $t^{(1-\alpha/2)} \ell(8/\sqrt{t})\leq 
    1$ for all $0\leq t < t_0$. So, \eqref{eq:d11_bound} is bounded 
    by
    \[
        \kappa \max{\left( u^{-\alpha-\delta}, u^{-\alpha+\delta} \right)} 
        \wedge 1 \in L^{1}([0,\infty) ).
    \]
    Thus, we can apply Lebesgue's Dominated Convergence Theorem to $D_{11}(t)$ 
    which gives $D_{11}(t) \rightarrow
    0$, as $t\rightarrow 0$, since $\Probs{L_{t} \geq \sqrt{t} u/2 } 
    \rightarrow 0$, as $t\rightarrow 0$,
    for $u>0$. We now consider $D_{12}(t)$ and estimate
    \begin{align}
        D_{12}(t) &= \int_{0}^{\infty} \Probs{L_{t} \geq \frac{\sqrt{t} 
        u}{2} } \Ind{\frac{8}{\sqrt{t} u} < M_0} du \notag \\
        &\leq \kappa t ^{1- \alpha/2} \int_{0}^{\infty} u^{-\alpha} \ell\left( 
        \frac{8}{\sqrt{t} u} \right) 
        \Ind{\frac{8}{\sqrt{t} u} < M_0} du \notag \\
        &= \kappa t^{1- \alpha/2} \int_{0}^{\infty} \left( \frac{8w}{\sqrt{t}}
        \right)^{-\alpha} \ell\left( \frac{1}{w} \right)
        \Ind{\frac{1}{w}< M_0} \frac{8 dw}{t^{1/2}} \notag \\
        &= \kappa t^{1/2} \int_{1/M_0}^{\infty} w^{-\alpha} \ell(1/w) dw 
        = \kappa t^{1/2} \int_{1/M_0}^{\infty} \gamma^{*}(z) dz \rightarrow 0,
        \notag
    \end{align}
    since the last integral is finite. Therefore, $\int_{0}^{\infty} 
    \Probs{L_{t} \geq \frac{\sqrt{t} u}{2} } du \rightarrow 0$ as 
    $t\rightarrow 0$. Applying Lemma \ref{EDCT} gives the result. 
\end{proof}

\begin{proof}[Proof of Theorem~\ref{thm:BrownianSimplified}]
    Note we still have \eqref{eq:ch_fn_convergence} and 
    \eqref{eq:prob_convergence}. Continuing as in the proof of 
    Theorem~\ref{thm:Brownian}, 
    \begin{align}
        \frac{1}{\sqrt{t}} c(t,0) &= \frac{1}{\sqrt{t}} \int_{0}^{\infty} e^{-x} 
        \Probs{X_t \geq x} dx
        = \int_{0}^{\infty} e^{-\sqrt{t} u} \Probs{X_t \geq u \sqrt{t}} du. 
        \label{eq:main_int3}
    \end{align}
    and
    \begin{align}
        e^{-\sqrt{t} u} \Probs{X_{t} \geq u \sqrt{t} } \leq \;\; 
        &\Probs{\sigma W_{t} \geq \frac{\sqrt{t} u}{2}} 
        + \Probs{L_{t} \geq \frac{\sqrt{t} u}{2} }
        = \Probs{\sigma W_{1} \geq \frac{u}{2}}  + \Probs{L_{t} \geq 
        \frac{\sqrt{t} u}{2} }. \notag 
    \end{align}

    We aim to use Lemma~\ref{EDCT} with 
    \begin{align*}
        f(u,t) &= e^{-\sqrt{t} u} \Probs{X_{t} \geq u \sqrt{t} }, \\
    \bar{f}(u) &= \Probs{\sigma W_{1} \geq u}, \\
    h(u) &= \Probs{\sigma W_{1} \geq u/2}, 
    \end{align*}
    and $g(u,t) = \Probs{L_{t} \geq \sqrt{t} u/2}$. It is easy to see that
    conditions \ref{C2}-\ref{C3} of the lemma are satisfied. We need to show 
    \ref{C5} and \ref{C6} hold. It is again not too difficult to see
    that \ref{C5} holds from \eqref{eq:ch_fn_convergence} and the 
    independence of $W$ and $L$ under $\Ps$. 
    We now show that \ref{C6} is satisfied as well.  Namely, we prove that
    \[
        \int_{0}^{\infty} \Probs{L_{t} \geq \frac{\sqrt{t} u}{2} } du 
        \rightarrow 0,
    \]
    as $t\rightarrow 0$. Once we show this, we immediately have the result 
    by applying Lemma~\ref{EDCT} to \eqref{eq:main_int3}. Note that there 
    exists $t_0 > 0$ such that $B_{t} \leq \sqrt{t}$ for all 
    $0 \leq t \leq t_{0}$. For $0\leq t \leq t_0$, we have
    \[
        \Probs{L_{t} \geq \sqrt{t} u} \leq \Probs{L_{t} \geq B_{t} u }.
    \]
    for every $u\geq 0$. 

    To simplify notation, let $F(t,u) = \Probs{L_{t} \geq \sqrt{t} 
    u/2}$, $G(t,u) =  \Probs{L_{t} \geq B_{t} u/2 }$, and
    $\bar{G}(u) = \Probs{Z \geq u}$. Note that $\int_{0}^{\infty} G(t,u)du 
    \rightarrow \int_{0}^{\infty} \bar{G}(u) du$ as $t\rightarrow 0$. It is 
    clear that $0 \leq \liminf_{t\rightarrow 0} \int_{0}^{\infty} F(t,u) du$. 
    Now, $G(t,u) - F(t,u) \geq 0$, so we can apply Fatou's lemma to get
    \begin{align}
        \int_{0}^{\infty} \bar{G}(u) du &\leq \liminf_{t \rightarrow 0} 
        \int_{0}^{\infty} \left( G(t,u) - F(t,u) \right) du \notag \\
        &=\liminf_{t\rightarrow 0} \left( \int_{0}^{\infty} G(t,u) du - 
        \int_{0}^{\infty} F(t,u) du \right)
        = \int_{0}^{\infty} \bar{G}(u) du - \limsup_{t\rightarrow 0} 
        \int_{0}^{\infty} F(t,u) du. \notag
    \end{align}
    Canceling terms gives $\limsup_{t\rightarrow 0} \int_{0}^{\infty} F(t,u) 
    du \leq 0$. Thus, 
    \[
        0 \leq \liminf_{t\rightarrow 0} \int_{0}^{\infty} \Probs{L_{t} \geq 
        \frac{\sqrt{t} u}{2}} du \leq
        \limsup_{t\rightarrow 0} \int_{0}^{\infty} \Probs{L_{t} \geq 
        \frac{\sqrt{t} u}{2}} du \leq 0,
    \]
    and therefore \ref{C6} is satisfied, proving the result. 
\end{proof}

\begin{proof}[Proof of Corollary~\ref{cor:brownian_case}]
    We proceed exactly as in the proof of Corollary \ref{cor:no_brownian_case}. 
    We are now comparing \eqref{eq:BS_asymptotics} with \eqref{eq:call_price_bs}
    multiplied by $S_0$. So,
    \[
        \frac{S_{0} \hat{\sigma}(t)}{\sqrt{2 \pi}} \sqrt{t} \sim S_{0} \sigma 
        \sqrt{t} \Expecsnb{\left( W_{1} \right)_{+}},
    \]
    as $t \rightarrow 0$. This implies that 
    \[
        \hat{\sigma}\left( t \right) \sim \sigma \sqrt{2 \pi} \Expecsnb{ \left( 
            W_{1}\right)_{+}},
    \]
    as $t \rightarrow 0$.
\end{proof}

\bibliography{levywork}

\end{document}